%% file: arxiv_version_v2.tex
\documentclass[a4paper,english,11pt]{article}
\usepackage[utf8]{inputenc}
\usepackage[T1]{fontenc}
\usepackage{babel}
\usepackage{graphicx}
\usepackage[left=80pt,
	    textwidth=440pt,
	    textheight=680pt,
	    top=90pt,
	    footskip=40pt]{geometry}

\usepackage{color}
\usepackage{amsmath}
\usepackage{amssymb}
\usepackage{amsthm}
\newtheorem{theorem}{Theorem}[section]
\newtheorem{lemma}[theorem]{Lemma}
\newtheorem{proposition}[theorem]{Proposition}

\theoremstyle{remark}

\theoremstyle{definition}
\newtheorem{definition}[theorem]{Definition}
\newtheorem{remark}[theorem]{Remark}
\newtheorem{notation}[theorem]{Notation}

\usepackage{multirow}
\usepackage{enumitem}
\usepackage{setspace}

\newcommand{\R}{\mathbb{R}}
\newcommand{\E}{\mathbb{E}}

\renewcommand{\vec}[1]{\boldsymbol{#1}}

\newcommand{\LEPTmu}{\ensuremath{\textrm{LEPT}_{\mathcal{F}}}}
\newcommand{\FLEPT}{\ensuremath{\textrm{LEPT}_{\mathcal{F}}}}
\newcommand{\LEPTLS}{\textrm{LEPT}_{\mathcal{P}}}

\newcommand{\OPT}{\ensuremath{OPT_\mathcal{P}}}

\newcommand{\PoFA}{\operatorname{PoFA}}

\newcommand{\PoNS}{\operatorname{PoNS}}

\def\clap#1{\hbox to 0pt{\hss#1\hss}}

\title{Stochastic Extensible Bin Packing\footnotemark[3]}

\author{Guillaume Sagnol\footnotemark[1] \ and
Daniel Schmidt genannt Waldschmidt\footnotemark[1]\\[1em]
{\small Technische Universität Berlin, Fakultät II, Institut für Mathematik,}\\
{\small MA 5-2, Straße des 17. Juni 136, 10623 Berlin, Germany.}\\
\texttt{ \{sagnol,dschmidt\}@math.tu-berlin.de}}
\date{}

\begin{document}

\maketitle

\renewcommand*{\thefootnote}{\fnsymbol{footnote}}
\footnotetext[1]{This research was funded by the Deutsche Forschungsgemeinschaft (DFG, German Research Foundation) under Germany's Excellence Strategy – The Berlin Mathematics Research Center MATH+ (EXC-2046/1, project ID: 390685689).}

\footnotetext[2]{ Some results of this article have been presented in a conference~\cite{SST18}, and we extend them in different directions. In particular, we prove a result which was only conjectured in the conclusion of~\cite{SST18}, cf.\ Section~\ref{sec:long}, and we give improved
performance guarantees for instances with 
processing times coming from given families of probability distributions, cf.\ Section~\ref{sec:stochastic_dominance}.}

\begin{abstract}
 We consider the \emph{stochastic extensible bin packing problem} (SEBP) in which $n$ items of stochastic size are packed into $m$ bins of unit capacity. In contrast to the classical bin packing problem, the number of bins is fixed and they can be extended at extra cost. This problem plays an important role in stochastic environments such as in surgery scheduling: Patients must be assigned to operating rooms beforehand, such that the regular capacity is fully utilized while the amount of overtime is as small as possible.
 
 This paper focuses on essential ratios between different classes of policies: First, we consider the price of non-splittability, in which we compare the optimal non-anticipatory policy against the optimal fractional assignment policy. We show that this ratio has a tight upper bound of $2$. Moreover, we develop an analysis of a fixed assignment variant of the LEPT rule yielding a tight approximation ratio of $(1+e^{-1}) \approx 1.368$. 
 Furthermore, we prove that the price of fixed assignments, related to the benefit of adaptivity, which describes the loss when restricting to fixed assignment policies, is within the same factor. This shows that in some sense, LEPT is the best fixed assignment policy we can hope for. We also provide a lower bound on the performance of this policy comparing against an optimal fixed assignment policy. Finally, we obtain improved bounds for the case where the processing times are drawn from a particular family of distributions, with either a bounded Pietra index or when the
 familly is stochastically dominated at the second order.
 
\end{abstract}


\section{Stochastic Extensible Bin Packing}

In the \emph{extensible bin packing problem} (EBP), we must put $n$ items of size $(p_1, \ldots,p_n)$ in $m$ bins, where the bins
can be extended to hold more than the regular unit capacity. The cost of a bin is its regular capacity together with its extension costs: Specifically,
a bin holding the items $I\subseteq\{1,\ldots,n\}$ has a cost of $\max\bigl(\sum_{i\in I} p_i,1\bigr)$.
The goal is to minimize the total cost of the $m$ bins.

The model of extensible bin packing naturally arises in scheduling problems with 
machines available for some amount of time at a fixed cost, and an additional cost for extra-time.
So we stick to the scheduling terminology in this article (bins are machines, items are jobs, and
item sizes are processing times).
Recently, the model of EBP was adopted 
to handle surgery scheduling problems~\cite{BD17,DMBH10,S+18}: here, the machines are operating rooms,
and the jobs are operations to be performed on elective patients. The extension of the regular working time of a machine
corresponds to overtime for the medical staff.
This application to surgery scheduling motivates the present paper: in practice, the duration of a surgical operation on a given patient
is not known with certainty. Therefore, we want to study the stochastic counterpart of the extensible bin packing problem, in which the processing durations $p_j$'s are only known probabilistically, and the
expected cost of the machines is to be minimized.

\medskip
\noindent\textbf{Related work.}
EBP is closely related to another scheduling problem,
where each job $j$ has a due date $d_j$ and the goal is to minimize the \emph{total tardiness} $\sum_j T_j$,
where $T_j$ is the positive part of the difference of its completion time and its due date.
This problem can not be approximated within any constant factor in polynomial time,
unless $P=NP$~\cite{KW02}. It relies on the fact that an approximation algorithm could differentiate YES and NO instances of PARTITION, since for YES instances the objective is equal to $0$. Therefore, several articles
studied approximation algorithms for a modified tardiness criterion,
$\sum T_j + d_j$;
see~\cite{KS07,LXCZ09}. The situation is very similar for extensible bin packing:
the problem of minimizing the amount by which bins have to be extended
is not approximable, and the criterion of EBP is obtained by adding the constant $m$ to the objective.

The (deterministic version of) EBP was introduced by~\cite{DKST98}, who showed that the problem is strongly NP-hard, by reducing from 3--PARTITION; cf.~\cite{GJ79}. Moreover, they prove that the 
\emph{longest processing time first} (LPT) algorithm
--which considers the jobs sorted in nonincreasing order of their processing time and assigns them sequentially to the machine with the largest remaining capacity--
is a $\frac{13}{12}-$approximation algorithm. There is also an FPTAAS (fully polynomial asymptotic approximation scheme) for EBP~\cite{CL01}. For equal bins, LPT can also be interpreted as iteratively assigning the jobs to the machine with the currently smallest load.
In~\cite{DS99} the LPT algorithm was shown to be a $2(2-\sqrt{2})\simeq1.1716-$approximation algorithm for the case of unequal bin sizes.
In a more general framework, Alon et al.\ present a polynomial time approximation scheme~\cite{AAWY98}.

The online version of the problem also attracted attention. Here, the
jobs arrive one at a time and they must be assigned to a machine irrevocably.
The list scheduling algorithm LS that assigns an incoming job to the machine with the largest remaining capacity was shown to have a competitive ratio
of $\frac{5}{4}$ for equal bin sizes in~\cite{DS99},
under the assumption that each job fits in one bin,
and was generalized in~\cite{ST99} for the case with unequal bin sizes.
Furthermore, it was proven
that no algorithm can achieve a performance of $\frac{7}{6}$ or smaller compared to the offline optimum.
An improved online algorithm with a competitive ratio of $1.228$ was also presented in~\cite{ST99}.

In the context of surgery scheduling, a slightly more general framework has been introduced in~\cite{DMBH10}:
the decision maker also chooses the number of bins of size $S$ to open, at a fixed cost $c^f$, and there
is a variable cost $c^v$ for each minute of overtime. It is 
observed in~\cite{BD17} that every $(1+\rho)-$approximation
algorithm for EBP yields a $(1+\rho \frac{Sc^v}{c^f})$-approximation algorithm in this more general
setting. They also consider a two-stage stochastic
variant of the problem, in which emergency patients should be allocated to operating rooms
with pre-allocated elective patients. For this problem (in the case $S=c^v=c^f=1$),
a particular fixed assignment policy was shown to be a $\frac{5\theta}{4}$-approximation algorithm,
when each job has a duration with bounded support $P_j\in[0,p_j^{\max}]$ such that $p_j^{\max}\leq\theta\E[P_j]$.
To the best of our knowledge, this has been the only attempt to consider stochastic jobs in the literature on EBP.

When considering stochastic optimization problems adaptive and non-adaptive policies are the solution concepts of matter. Especially, the greatest ratio between the cost of an optimal non-adaptive and the cost of an optimal adaptive policy over all instances is a quantity of interest. This so-called \emph{benefit of adaptivity} or \emph{adaptivity gap} has drawn attention dating back to the work in \cite{DGV08} and is getting popular, see e.g. \cite{BN15,DGV05,GNS16}. In this work, we will work with another slightly different ratio closely related to it, since in the field of stochastic scheduling we are concerned with non-anticipatory policies that can make time-dependent decisions, such as idling. This can make a difference in the setting of parallel machines. 

In stochastic scheduling problems various notions of stochastic dominance have been considered to obtain optimal policies for specific classes of processing time distributions; see e.g. the book by Pinedo~\cite{Pin16} and the references therein. 
We use the notions of second-order stochastic dominance and Lorenz dominance in this work.
In addition, several approximative policies have been designed where the performance guarantee is parameterized by some coefficient measuring the dispersion of the random processing times: For instance, Uetz~\cite{Uet01} used the coefficient of variation of a random variable and Megow, Uetz and Vredeveld~\cite{MUV06} introduced the notion of $\delta$-NBUE processing times. In our work we consider the Pietra index as well as the Gini index. To the best of our knowledge, this is the first work to obtain approximative policies using these indices or these notions of stochastic dominance in this context.

\medskip
In the remaining of this section, we 
introduce the \emph{stochastic extensible bin packing problem} (SEBP).
Throughout, we consider the (offline) problem of scheduling $n$ stochastic jobs on $m$ parallel identical machines non-preemptively, where $n>m$ as the problem is trivial otherwise. 
We will assume that the distribution of the processing times are given beforehand and that their expectation is finite and computable\footnote{
We do not specify how the processing time distributions should be represented in the input of the problem,
as the policies we study only require the expected value of the processing times. In fact, we could even assume
a setting
in which the input consists only of the mean processing times $\mu_j=\E[P_j]$ ($\forall j\in\mathcal{J}$), and an adversary
chooses some distributions of the $P_j$'s matching the vector $\vec{\mu}$ of first moments.
}.\label{footnote:DRO}
The set of machines and jobs are denoted by $\mathcal{M}=\{1,\ldots,m\}$  and $\mathcal{J}=\{1,\ldots,n\}$, respectively. 

\medskip
\noindent \textbf{Stochastic Scheduling.}
Now, we want to give the intuition and main ideas of the required background in the field of stochastic scheduling. Precise definitions are given in~\cite{MRW84}. The processing times are represented by a vector
$\vec{P}=(P_1,\ldots,P_n)$ of random variables. We denote by $\vec{p}=(p_1,\ldots,p_n)\in\R^n_{\geq 0}$ 
a particular realization of $\vec{P}$. We assume that the $P_j$'s are mutually independent,
and that each processing time has a finite expected value.
Unlike the deterministic case, a scheduling strategy can take more general forms
than just an allocation of jobs to machines, as information is gained during the execution
of the schedule. Indeed, job durations become known upon completion, and
adaptive policies can react to the processing times observed so far.

We define a \emph{schedule} as
a pair $S=(\vec{s},\vec{a})\in \R_{\geq 0}^n \times \mathcal{M}^n,$
where $s_j\geq 0$ is the starting time of job $j$ and $a_j\in\mathcal{M}$ is the machine to which job $j$ is assigned.
A schedule $S$ is said to be \emph{feasible} for the realization $\vec{p}$
if each machine processes at most one job at a time:
$$
\forall i\in\mathcal{M},\ \forall t\geq 0,\quad \Big|\{ j\in\mathcal{J}:\  a_j=i,\ s_j\leq t < s_j + p_j \}\Big| \leq 1. 
$$
We denote by $\mathcal{S}(\vec{p})$ the set of all feasible schedules for the realization $\vec{p}$.
A \emph{planning rule} is a function $\Pi$ that maps a vector $\vec{p}\in\R^n_{\geq 0}$ 
of processing times to a schedule $S\in\mathcal{S}(\vec{p})$. A planning rule
is called a \emph{scheduling policy} if it is \emph{non-anticipatory}, which intuitively means that decisions taken at time $t$ (if any) may only depend
on the observed durations of jobs completed before $t$, and the probability distribution of the other processing times (conditioned by
the knowledge that ongoing jobs have not completed before $t$).

\medskip
\noindent \textbf{Stochastic Extensible Bin Packing (SEBP).}
For a scheduling policy $\Pi$, we denote by
$S_j^\Pi$ and $A_j^\Pi$ the random variables for the starting time
of job $j$, and the machine to which $j$ is assigned, respectively.
The completion time of job~$j$ is 
$C_j^\Pi=S_j^\Pi+P_j$.
We further introduce the random variable $X_i^\Pi$ for the workload of machine $i$,
which is defined as the latest completion time of a job on machine $i$:
$$
X_i^\Pi := \max \{C_j^\Pi|\ j\in\mathcal{J},\ A_j^\Pi = i\}. 
$$
It is easy to see that when $\Pi$ is \emph{non-idling}, i.e., if the starting time
of any job is either $0$ or equal to the completion time of the previous job assigned to the same machine, then
$$X_i^\Pi = \sum_{\{j\in\mathcal{J}:\ A_j^\Pi=i\}} P_j.$$
The realizations of
the random vectors $S^\Pi,A^\Pi,C^\Pi$ and $X^\Pi$
for a vector of processing times $\vec{p}$ are denoted by appending $\vec{p}$
as an argument. For example, the workload of machine $i$ for a non-idling policy $\Pi$
in the scenario $\vec{p}\in\R^n_{\geq 0}$ is
$$X_i^\Pi(\vec{p}) = \sum_{j\xrightarrow{\Pi(\vec{p})} i} p_j,$$
where $j\xrightarrow{\Pi(\vec{p})} i$ means that $\Pi(\vec{p})$
assigns job $j$ to machine $i$,
i.e., we sum over indices $\sloppy{\{j\in\mathcal{J}: A_j^\Pi(\vec{p})=i\}}$.

\smallskip
We assume that jobs are scheduled on machines with an extendable working time,
each machine having a unit regular working time.
The cost incurred on machine $i$ is equal to
$\max(X_i^\Pi,1)$,
which accounts for the fixed costs, plus the amount by which the regular working time has to be extended.
We are interested in strategies that
minimize the expected value of the total costs
\begin{equation*}
\Phi(\Pi) := \E\Big[\, \sum_{i\in \mathcal{M}} \max(X_i^\Pi,1)\, \Big]. 
\end{equation*}
The objective can also be defined realization-wise 
$\phi\big(\Pi,\vec{p}\big):=\sum_{i\in\mathcal{M}} \max(X_i^\Pi(\vec{p}),1)$,
so that $\Phi(\Pi) = \E[\phi\big(\Pi,\vec{P}\big)]$.

\begin{remark}
	Approximation results for the Stochastic Extensible Bin Packing Problem can easily be extended to the more general two-stage problem introduced in~\cite{DMBH10} where we additionally have to decide beforehand on the number of machines to use. For this purpose, we use for each integer $m\leq n$ our policy and select the number of bins for which we obtain the smallest objective value obtaining the same performance guarantee.
\end{remark}

\bigskip
\noindent \textbf{Classes of scheduling policies.}
We define the following classes of scheduling policies:
\begin{itemize}
 \item $\mathcal{P}$ denotes the class of all scheduling policies (non-anticipatory planning rules).
 \item $\mathcal{F}$ denotes the set of all \emph{non-idling fixed-assignment policies}.
 Such policies are characterized by a vector of job-to-machine assignments $\vec{a}\in\mathcal{M}^n$, so
 that $A^\Pi(\vec{p})=\vec{a}$ does not depend on the realization of processing times. 
 For such a policy $\Pi$, it holds
 \begin{equation*}
 \Phi(\Pi)=\sum_{i\in\mathcal{M}} \mathbb{E}\Big[\max(\sum_{j\xrightarrow{\Pi} i} P_j,1)\Big], 
 \end{equation*}
 where the sum indexed by ``\ $j\xrightarrow{\Pi} i$\ '' goes over all jobs $j$ such that $A_j^\Pi=i$. 
\end{itemize}

The distinction between fixed assignment policies and other, more sophisticated adaptive policies plays a central role in this article.
Indeed, in the context of surgery scheduling, committing to a fixed assignment policy is a common practice~\cite{BD17,DMBH10,S+18},
because fixed assignments yield simple schedules, that are easier to apprehend for both the medical staff and the patients.
Hence, they cause less stress
and are better suited to handle
the human resources of an operating theatre~\cite{DT02}.
Nonetheless, there is currently active research on the use of reactive policies for operating room scheduling~\cite{ZXG14}.
As ``\emph{fully adaptive scheduling models and policies are infeasible in operating room scheduling practice}'',
the focus is now on hybrid scheduling policies with a large amount of static decisions, and a limited amount of adaptivity~\cite{XJDK18}.
While more flexible policies could arguably lead to an important gain of efficiency over static policies, there are still
many obstacles for their introduction in the operating theatre. In particular, it must be ensured
that adaptive policies do not harm the quality of health care~\cite{ZYLLSLY16}, and computer-assisted scheduling techniques need to gain acceptance
among practitioners~\cite{ISM10}.
In this context, one goal of the present paper is to study the gap between fixed assignment and adaptive policies from a theoretical perspective.

\bigskip
In addition, we define the following class of fractional policies, which is related to scheduling problems concerning moldable work preserving tasks (see \cite{Leu04}). It cannot be considered as non-anticipatory planning
rules, but will be useful to derive bounds:
\begin{itemize}
 \item $\mathcal{R}$ denotes the class of fractional assignment policies, 
 in which a fraction $a_{ij} \in [0,1]$ of job~$j$ is to be executed on machine $i$, with $\sum_{i\in \mathcal{M}} a_{ij} = 1$, for all $j\in\mathcal{J}$. 
 For a ``policy'' $\Pi\in\mathcal{R}$, the different fractions of a job can be executed simultaneously on different machines, so
 \begin{equation*}
   \Phi(\Pi)=\sum_{i\in\mathcal{M}} \mathbb{E}\Big[\max(\sum_{j\in \mathcal{J}} a_{ij}^{\Pi}\, P_j,1)\Big].
 \end{equation*}
\end{itemize}

\noindent \textbf{LEPT policies.}
There is no unique way to generalize the LPT algorithm used in the deterministic case.
We distinguish two variants
of the ``longest expected processing time first'' (LEPT) policy.
The policy $\LEPTmu$ is the fixed assignment policy that 
results in the same assignments as the LPT algorithm for the deterministic processing times $p_j=\mathbb{E}[P_j]$.
In other words, job to machine assignments are precomputed offline, as follows:
jobs are considered in decreasing order of $\mathbb{E}[P_j]$, and sequentially assigned to the least loaded machine (in expectation). An example of $\LEPTmu$ is depicted in Figure~\ref{figure}.
The second policy, which we denote by $\LEPTLS$, is the static list policy
which considers jobs in the order of decreasing $\mathbb{E}[P_j]$'s,
and start them (in this order) as early as possible. Unlike $\LEPTmu$, the job to machine assignments
of the list policy $\LEPTLS$ 
depend on the realization $\vec{p}$ of the processing times. By~\cite{ST99} it immediately follows that $\LEPTLS$ is a $\frac{5}{4}$-approximation with respect to $OPT_\mathcal{P}$ for the case of short jobs only (cf. Definition~\ref{def:short}), since in every realization the schedule produced by $\LEPTLS$ is obtained by list scheduling.

As discussed earlier, given the prominence of fixed assignment policies in the context of surgery scheduling, 
we focus on the policy $\LEPTmu$ in the remaining of this article.

\begin{figure}
 \includegraphics[scale=0.8]{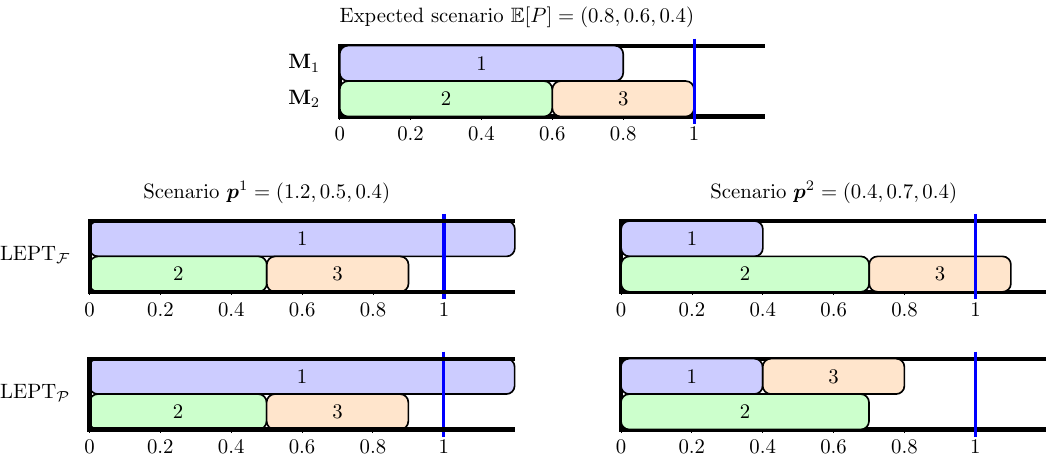} 
 \caption{{\small Example of a fixed assignment policy: assume machines $\mathcal M = \{1,2\}$ and jobs $\mathcal J = \{1,2,3\}$ with processing time distributions $p_1 \in \{0.4,1.2\}$, $p_2 \in \{0.5,0.7\}$, $p_3=0.4$ where the duration of each stochastic job is attained with probability $\frac{1}{2}$. Since $\mathbb{E}[P_1] = 0.8 \geq \mathbb{E}[P_2] = 0.6 \geq p_3 = 0.4$,
 $\LEPTmu$ assigns the jobs in order $1 \rightarrow 2 \rightarrow 3$ to the machines before their realization is known.
 The figure on the top depicts the resulting job to machine assignments with the average durations. 
 For the realization $\vec{p}_1 = (1.2,0.5,0.4)$ (lower left), $\LEPTmu$ is optimal with cost $2.2$.
 For the realization $\vec{p}_2 = (0.4,0.7,0.4)$ (lower right), $\LEPTmu$ yields cost $2.1$.
 In contrast, note that $\LEPTLS$ would have started job~$3$ on the first machine after completion of job~$1$, giving a cost of $2$.\label{figure}}}
\end{figure}

\medskip
\noindent \textbf{Performance ratios.}
For a given instance $I=(\vec{P},m)$ of the SEBP,
we denote the optimum value in the class $\mathcal{C}$ of scheduling policies 
by
\begin{equation*}
 OPT_{\mathcal{C}}(I) = \inf_{\Pi\in\mathcal{C}}\ \Phi(\Pi).
\end{equation*}
Note that $\Phi$ is continuous, in particular lower semi-continuous, and hence by Theorem 4.2.6 of \cite{MRW84} there exists an optimal policy. Whenever the instance is clear from the context, 
or when $I=(\vec{P},m)$ is an arbitrary instance, we will drop $I$ from the argument, so we simply write $OPT_\mathcal{C}$.
We also denote by $OPT(\vec{p})$ the optimal objective value for the deterministic problem with processing times $\vec{p}$.
In this case, it is clear that we can restrict our attention to fixed assignment policies $\Pi\in\mathcal{F}$:
\begin{equation*}
OPT(\vec{p}) = \inf_{\Pi\in\mathcal{F}}\ \phi(\Pi,\vec{p}). 
\end{equation*}
For notational convenience, we will abuse notation and write $OPT_{\mathcal{C}}$, $\LEPTmu$, $\Pi$ and $OPT(\vec{p})$ to denote both the policy as well as the objective value obtained by the policies. Furthermore, we denote the expected value of an optimal anticipative policy by $\E[OPT(\vec{P})]:=\E[\phi(OPT(\vec{P}),\vec{P})]$. 
Let us now define various performance ratios.
We say that $\Pi\in\mathcal{C}$ is a \emph{$\gamma$-approximation in the class $\mathcal{C}$} if
the inequality
$ \Phi(\Pi) \leq \gamma\, OPT_\mathcal{C}$
holds for all instances of SEBP.
The \emph{price of fixed assignments} and the \emph{price of non-splittability} are respectively defined by
\begin{equation*}
\PoFA = \sup_I \frac{OPT_{\mathcal{F}}(I)}{OPT_{\mathcal{P}}(I)}\qquad\quad \text{and} \quad \qquad  \PoNS = \sup_I \frac{OPT_{\mathcal{P}}(I)}{OPT_{\mathcal{R}}(I)}, 
\end{equation*}
where the suprema go over all instances $I$ of SEBP.

The first ratio ($\PoFA$) describes the loss if we restrict our attention to fixed assignment policies. In other words, it is a measure of 
what can be gained by allowing the use of more flexible, adaptive policies. This quantity gained attention in classical scheduling problems, e.g., in~\cite{MUV06} and~\cite{SSU16}, whereby the latter shows that it can be arbitrarily large for the objective of minimizing the expected sum of completion times on parallel identical machines as the coefficient of variation grows.

The second ratio ($\PoNS$) 
is related to the power of preemption,
see e.g.~\cite{CI98,CSV12,SS02,SS14}, but should not be mixed up with it, 
because the class $\mathcal{R}$ allows different parts of a job to be processed simultaneously
on several machines for fractional assignment policies.
However, this quantity has a simple interpretation in the context of surgery scheduling.
Consider a hospital that assigns patients to a particular day until the total expected duration
of the booked surgeries exceeds a certain threshold, but ignores the actual allocation of patients to operating rooms.
The precise assignment of patients to operating rooms is deferred to a later stage, typically one week to one day
before the day of surgery, when the set of all elective patients will be known.
In fact, this simplification amounts to assuming that jobs of a particular day are placed in a single bin of size $m$ (rather than in $m$ bins of unit size).
We will see in Proposition~\ref{prop:order} that this can be interpreted as splitting the patient durations arbitrarily,
and hence, evaluating the costs within this simplified one-bin model can yield a multiplicative error of up to $\PoNS$.

\medskip
\noindent\textbf{Organization and Main results.}
Our paper is organized as follows. Section~\ref{sec:pons} deals with the price of non-splittability. We show that the expected cost
of an optimal
non-anticipatory policy is at most twice the expected cost of an optimal fractional assignment policy.
Moreover, we present instances that achieve a lower bound arbitrarily close to $2$, showing that $\PoNS = 2$. In Section~\ref{sec:1+e-1},
we consider the case of short jobs ($P_j \in [0,1]$ almost surely) 
and we obtain a performance guarantee of $1+e^{-1}$ for $\LEPTmu$ compared to the stochastic optimum.\footnote{We decided to include this special case as it is not as technical as the general case, it is a common and reasonable assumption on the distribution of the processing times~\cite{DS99,ST99}, and we reuse some results of this part in later sections.}
This result is generalized in the next section for long jobs.
In Section~\ref{sec:pofa} we show 
that the price of fixed assignments is at most $1+e^{-1}$.
We also give a family of instances where this bound is attained at the limit,
which proves that $\PoFA = 1+e^{-1}$.
This shows that $\LEPTmu$ is --in a certain sense-- the best possible fixed assignment policy.
The next section shows that the performance of $\LEPTmu$ can not be better than $\frac{4}{3}$ in the
class~$\mathcal{F}$. In Section~\ref{sec:stochastic_dominance} we give improved upper bounds when the processing time come from a specific family of distributions.
In particular,
we prove that the \FLEPT\ policy
is $1.2334$-approximative when the 
processing times are
lognormally distributed
with a squared coefficient of variation $\Delta\leq\frac{1}{4}$, 
a reasonable assumption in the context of surgery scheduling~\cite{strum1998surgical,S+18}.
Other authors suggested to use Gamma or Weibull
distributed random variables~\cite{joustra2013can,choi2012analysis},
or $3$-parameters lognormal distribitions~\cite{stepaniak2009modeling}
to approximate surgery durations; our approach can also handle these cases, see Table~\ref{tab:approx}. Finally, we give a distribution-free
bound for instances with a bounded Pietra or Gini index in Theorem~\ref{theo:pietra}.


\section{The price of non-splittability}\label{sec:pons}

From now on we fix some instance $(\vec{P},m)$. We begin with a convenient notation followed by a basic proposition.

\begin{notation}
	Let $\rho$ denote the expected workload averaged over all machines, in particular, $$\rho:=\frac{1}{m}\sum_{j\in\mathcal{J}} \E[P_j].$$
\end{notation}

\begin{proposition}\label{prop:order}
The following chain of inequalities holds:
 \begin{equation*}
 OPT_{\mathcal{F}} \geq OPT_{\mathcal{P}} \geq \E[OPT(\vec{P})] \geq OPT_{\mathcal{R}} = \E\Big[\max(\sum_{j\in\mathcal{J}} P_j, m)\Big] \geq m \max(\rho,1). 
 \end{equation*}
 \end{proposition}
\begin{proof}
The first inequality follows immediately since $\mathcal{F} \subseteq \mathcal{P}$. 

Next, for all policies $\Pi \in \mathcal{P}$ and all realizations $\vec{p}$
it holds 
$
  \phi(\Pi,\vec{p})\geq OPT(\vec{p}),
$
by definition of an optimal policy for the deterministic processing times $\vec{p}$.
Taking the expectation on both sides yields 
the second inequality.

Before we go on to the next inequality, we first show that 
$OPT_{\mathcal{R}} = \E[\max(\sum_{j\in\mathcal{J}} P_j, m)]$.
To do so we show that for any realization $\vec{p}$ an optimal fractional assignment policy assigns all jobs uniformly to all machines. 
More precisely, we show that $a_{ij}=\frac{1}{m}$ for all $i\in \mathcal{M}$ and $j\in \mathcal{J}$ solves the following
problem of finding the optimal fractional assignment:
\begin{equation}\label{fracP}
 \underset{0\leq a_{ij}\leq 1}{\vec{\operatorname{minimize}}} \quad\sum_{i\in \mathcal{M}}\max(\sum_{j\in\mathcal{J}} a_{ij} p_j, 1),\qquad \text{such that} \quad\sum_{i\in \mathcal{M}} a_{ij} = 1, \quad \forall j\in \mathcal{J}.
\end{equation}
 A trivial lower bound on the optimal value of Problem~\eqref{fracP} is $\max(\sum_{j\in\mathcal{J}} p_j, m)$. This is true since for any feasible fractional assignment $(a_{ij})_{i\in\mathcal{M},j\in\mathcal{J}}$,
 $\sum_{i\in \mathcal{M}}\max(\sum_{j\in\mathcal{J}} a_{ij} p_j, 1)\geq \sum_{i\in \mathcal{M}} \sum_{j\in\mathcal{J}} a_{ij} p_j = \sum_{j\in\mathcal{J}} p_j$,
 and similarly,
 $\sum_{i\in \mathcal{M}}\max(\sum_{j\in\mathcal{J}} a_{ij} p_j, 1)\geq \sum_{i\in \mathcal{M}} 1 = m$. 
 Choosing all fractions to be $\frac{1}{m}$ we obtain 
 $
  \sum_{i\in \mathcal{M}} \max(\sum_{j\in\mathcal{J}} \frac{1}{m} p_j, 1) = m\cdot \max(\sum_{j\in\mathcal{J}} \frac{1}{m} p_j, 1) = \max(\sum_{j\in\mathcal{J}} p_j, m)
 $
 which exactly matches the lower bound and hence, it must be optimal. Since this holds for any realization we can take the expected value resulting into the desired identity.

 In order to show $\E[OPT(\vec{P})] \geq OPT_{\mathcal{R}}$, we observe that for any realization $\vec{p}$,
 Problem~\eqref{fracP} is the continuous relaxation of the problem with binary variables for finding the optimal 
 assignments for the deterministic problem with processing times $\vec{p}$. Hence, by again taking expectations this yields the inequality. 
 
 \sloppy{Finally, the last inequality is Jensen's inequality applied to the convex function \mbox{$x\mapsto\max(x,m)$}.} 
\end{proof}

In the next proposition, we show the intuitive fact that among non-idling policies, the worst case
is to assign all jobs to the same machine.
\begin{proposition}\label{prop:naive}
Let $\Pi\in\mathcal{P}$ be non-idling and let $\Pi_1$ be the fixed assignment policy that schedules all jobs on
machine $1$. Then, we have $\Pi \leq \Pi_1$.
\begin{figure}
	\includegraphics[scale=0.8]{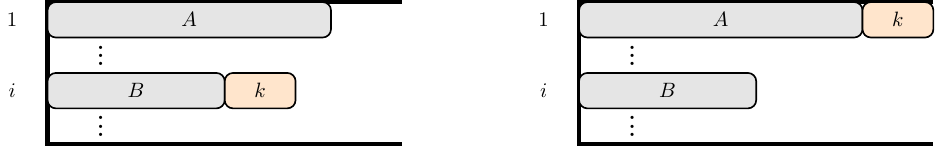} 
	\caption{{\small Illustration of the proof of Proposition~\ref{prop:naive}.\label{nonidling}}}
\end{figure}
\end{proposition}

\begin{proof}[Proof of Proposition~\ref{prop:naive}]
	To prove this result, we examine the change in the objective value of $\Pi$ when we move one job
	to the machine with highest load in $\Pi$, for a realization $\vec{p}$ of the processing times.
	W.l.o.g.\ let machine $1$ be the one with highest workload in $\Pi(\vec{p})$. Consider another machine $i\in \mathcal{M}\setminus \{1\}$ on which at least one job is scheduled.
	Let $k$ be the last job on machine $i$, 
	i.e., $C_k^\Pi(\vec{p})=X_i^\Pi(\vec{p})$. For the sake of simplicity, we define
	$A:= \{j\in \mathcal{J}|j\xrightarrow{\Pi(\vec{p})} i\}\setminus \{k\}$ and  $B:= \{j\in \mathcal{J}|j\xrightarrow{\Pi(\vec{p})} 1\}$.
	We consider another schedule $\Pi'(\vec{p})$ which coincides with $\Pi(\vec{p})$ except that job $k$ is scheduled on machine $1$ right after all jobs in $B$.
	We obtain
	\begin{alignat*}{2}
	& &&\phi(\Pi,\vec{p})- \phi(\Pi',\vec{p})\\
	& = && \max(\sum_{j\in A} p_j + p_k, 1) + \max(\sum_{j\in B} p_j, 1) - \Bigl(\max(\sum_{j\in A} p_j, 1) + \max(\sum_{j\in B} p_j + p_k, 1)\Bigr)\\
	& = && \begin{cases}
	1 + \max(\displaystyle{\sum_{j\in B}} p_j, 1) - \Bigl(1 + \max(\sum_{j\in B} p_j + p_k, 1)\Bigr) \quad &\ \text{if } \displaystyle{\sum_{j\in A}} p_j + p_k\leq 1\\
	\displaystyle{\sum_{j\in A}} p_j + p_k + \sum_{j\in B} p_j - \Bigl(\max(\sum_{j\in A} p_j, 1) + \sum_{j\in B} p_j + p_k\Bigr) \quad &\ \text{otherwise }
	\end{cases}\\
	& \leq && 0.
	\end{alignat*}
	
	Hence, iteratively moving some job $k$ to the fullest machine yields
	$\phi(\Pi,\vec{p})\leq\phi(\Pi_1,\vec{p})$. 
	Finally, the result follows by taking the expectation.
\end{proof}

We now prove that any non-idling policy is a 2-approximation in the class of non-anticipatory policies (and hence in the class of fixed-assignment policies).

\begin{proposition}\label{prop:naive2}
 Let $\Pi$ be \emph{any} non-idling policy. Then, 
 \begin{equation*}
  \Pi \leq 2\, OPT_{\mathcal{R}}.
 \end{equation*}
\end{proposition}
\begin{proof}

  Let $\Pi$ be a non-idling policy and $\Pi_1$ be the naive fixed assignment policy in which all jobs are scheduled on one machine without idle time.
  Proposition~\ref{prop:naive} yields
  that $\Pi\leq\Pi_1$, and we have
  \begin{equation*}
  \Pi_1
  = \E[\max(\sum_{j\in \mathcal{J}} P_j,1)]+(m-1)\leq \E[\max(\sum_{j\in \mathcal{J}} P_j,m)] + m -1. 
  \end{equation*}
  We know that $m\leq\E[\max(\sum_{j\in \mathcal{J}} P_j,m)]=OPT_\mathcal{R}$ from Proposition~\ref{prop:order}, so we have
  $$\Pi\leq\Pi_1 \leq 2\, OPT_\mathcal{R} -1\leq 2\, OPT_\mathcal{R}.$$
\end{proof}

Consequently, we are only interested in finding $\alpha-$approximation algorithms for $\alpha<2$, since a $2-$approximation algorithm
performs no better (in the worst case) than the naive policy that puts all jobs on a single machine. Before we state the main result, we need the next technical Lemma.
\begin{lemma}\label{lemm:poisson-max}
 Let $Y\sim$ Poisson($\lambda$) for some $\lambda\in\mathbb{N}$. Then,
 \begin{equation*}
  \frac{1}{\lambda}\E[\max(Y, \lambda)] = 1 + \frac{e^{-\lambda}\lambda^\lambda}{\lambda!}.
 \end{equation*}
\end{lemma}

\begin{proof}
	The proof simply works by exploiting the analytical form of Poisson probabilities:
	\begin{alignat*}{2}
	\frac{1}{\lambda}\E[\max(Y, \lambda)]& &&= \frac{1}{\lambda}\sum_{k=0}^{\infty}\max(k, \lambda)\cdot \frac{e^{-\lambda}\lambda^k}{k!}\\
	& &&= \frac{1}{\lambda}\sum_{k=0}^{\infty} k\cdot \frac{e^{-\lambda}\lambda^k}{k!} + \frac{1}{\lambda}\sum_{k=0}^{\infty} \max(0, \lambda - k)\cdot \frac{e^{-\lambda}\lambda^k}{k!}\\
	& &&= 1 + \sum_{k=0}^{\lambda} \Bigl(1 - \frac{k}{\lambda}\Bigr)\cdot \frac{e^{-\lambda}\lambda^k}{k!}\\
    & &&= 1 + e^{-\lambda}\cdot \Bigl(\sum_{k=0}^{\lambda} \frac{\lambda^k}{k!} - \sum_{k=1}^{\lambda} \frac{\lambda^{k-1}}{(k-1)!}\Bigr)\\
	& &&= 1 + \frac{e^{-\lambda}\lambda^\lambda}{\lambda!},
	\end{alignat*}
	where the last step follows from the property of a telescoping sum.
\end{proof} 

The last proposition also shows that the price of non-splittability is upper bounded by $2$. In fact, this bound is tight.
\begin{theorem}\label{theo:VoP2}
	The price of non-splittability of SEBP is $2$.
\end{theorem}

\begin{proof}[Proof of Theorem~\ref{theo:VoP2}]
It follows from Propositions~\ref{prop:order} and~\ref{prop:naive2} that $OPT_\mathcal{P} \leq OPT_\mathcal{F} \leq 2 OPT_\mathcal{R}$. 
 
Let $\lambda\in\mathbb{N}$ and consider the instance $I$ with $n=m\geq\lambda$ independent and identically distributed jobs in which the processing time
of each job $j$ takes the value $\frac{m}{\lambda}$ with probability $\frac{\lambda}{m}$ and $0$ otherwise. In other words, for all $j\in\mathcal{J}$
we have
 $P_j \sim \frac{m}{\lambda}\operatorname{Bernoulli}(\frac{\lambda}{m}).$
As $n=m$, an optimal non-idling policy clearly assigns each job to a different machine. 
This yields
 \begin{equation*}
	 OPT_{\mathcal{P}}(I) = m\cdot\E[\max(P_1,1)] = m\cdot\Bigl((1-\frac{\lambda}{m})\cdot 1 + \frac{\lambda}{m}\cdot\frac{m}{\lambda}\Bigr) = 2 m-\lambda.  
 \end{equation*}

 For the objective value of an optimal fractional assignment policy we can use Proposition~\ref{prop:order}. 
 We will also use the fact that the sum of i.i.d.\ Bernoulli random variables is binomially distributed, i.e.,
 $Y:=\frac{\lambda}{m}\cdot\sum_{j\in \mathcal{J}} P_j \sim \operatorname{Binomial} (m,\frac{\lambda}{m})$. 
 Moreover, it is folklore that $Y$ converges in distribution to $Z\sim$ Poisson($\lambda$) as $m \to \infty$.

Therefore, we have $\frac{\lambda}{m}OPT_{\mathcal{R}}(I) = \frac{\lambda}{m}\E\Big[\max(\sum_{j\in\mathcal{J}} P_j, m)\Big]=\mathbb{E}[\max(Y,\lambda)]$,
which converges in distribution to $1+\frac{e^{-\lambda} \lambda^\lambda}{\lambda !}$ as $m\to\infty$ by Lemma~\ref{lemm:poisson-max}.
Putting everything together, the ratio $OPT_{\mathcal{P}}(I)/OPT_{\mathcal{R}}(I)$ converges to $
2 (1+\frac{e^{-\lambda} \lambda^\lambda}{\lambda !})^{-1}$ as $m\to\infty$, and this quantity can be made arbitrarily close to $2$ by choosing $\lambda$ large enough.
\end{proof}

\section{Approximation ratio of LEPT: The case of short jobs}\label{sec:1+e-1}

In this section, we show that $LEPT_\mathcal{F}$ is an $(1+e^{-1})$-approximation algorithm
when the instance only contains \emph{short jobs}.
\begin{definition}\label{def:short}
	We say job $j$ is \emph{short} if its processing time $P_j$ is less than or equal to $1$ almost surely, i.e.,
	 \begin{equation*}
	 		\mathbb{P}[ 0\leq P_j\leq 1]=1.
	 \end{equation*}
\end{definition}

In order to prove the performance guarantee of $LEPT_\mathcal{F}$ we make use of several lemmas.
The first lemma gives a tight bound on the expected cost incurred on one machine.

\begin{lemma}\label{lemm:twopoints}
Let $k$ be some positive integer and let all jobs $j\in[k]$ be short.
Then,
\begin{equation*}
 \E\Bigl[\max(\sum_{j=1}^k P_j,1)\Bigr] \leq \sum_{j=1}^k \mathbb{E}[P_j]\ +\ \prod_{j=1}^k (1- \mathbb{E}[P_j]).
\end{equation*}
 Moreover, this bound is tight and attained for
 the two point distributions $P_j^* \sim \operatorname{Bernoulli}(\mathbb{E}[P_j])$.
\end{lemma}

\begin{proof}
	
	Let $Y$ and $Z$ be random variables with $\mathbb{P}[ 0\leq Y\leq 1]=1.$ Observe that $0\leq \E[Y]\leq 1$.
	We are going to show that $\E[\max(Y+Z,1)]$ can be bounded from above by choosing
	the two point distribution $Y^*\sim \operatorname{Bernoulli}(\E[Y])$, such that 
	$\mathbb{P}[Y^*=0] = (1-\E[Y])$  and $\mathbb{P}[Y^*=1] = \E[Y].$
	To do so, we define the function $g:[0,1]\to \R$, 
	$y\mapsto \E_Z[\max(y+Z,1)].$
	This function is convex, since it is the
	expectation of a pointwise maximum of two affine functions~\cite{BV04}.
	Therefore, for all $y\in [0,1]$ we have
	$
	g(y) \leq g(0) + y (g(1) - g(0)).
	$
	Then, by definition of $g$,
	\begin{align*}
	\E[\max(Y+Z,1)] = \E_Y[g(Y)] &\leq g(0) + \E_Y[Y]\cdot(g(1) - g(0))\\
	&= \E_{Y^*}[g(Y^*)]=\E[\max(Y^*+Z,1)]. 
	\end{align*}
	Using this bound for all $j\in[k]$,
	we obtain
	$\E\Bigl[\max(\sum_{j=1}^k P_j,1)\Bigr] \leq \E\Bigl[\max(\sum_{j=1}^k P_j^*,1)\Bigr],$
	where $P_j^* \sim \operatorname{Bernoulli}(\mathbb{E}[P_j])$. Then, 
	by the law of total expectation, we have:
	\begin{align*}
	\E\Bigl[\max(\sum_{j=1}^k P_j^*,1)\Bigr] &= \E\Bigl[\sum_{j=1}^k P_j^* \Big|\sum_{j=1}^k P_j^*\geq 1\Bigr] \mathbb{P}[\sum_{j=1}^k P_j^*\geq 1] + \mathbb{E}[1]\ \mathbb{P}[\sum_{j=1}^k P_j^*< 1].
	\end{align*}
	Since the random variable $\sum_{j=1}^k P_j^*$ is a nonnegative integer,
	it cannot lie in the interval $(0,1)$,
	so the first term in the above sum is equal to $\E\Bigl[\sum_{j=1}^k P_j^*\Bigr]=\sum_{j=1}^k \E\Big[P_j\Big]$,
	and the second term is equal to $\mathbb{P}[P_1^*=\ldots=P_k^*=0]=\prod_{j=1}^k (1-\mathbb{E}[P_j])$.
\end{proof}

Before we proceed, we introduce further notation on the outcome of $\LEPTmu$.
\begin{notation}
	Let $J_i:=\{j\in \mathcal{J} : j\xrightarrow{\LEPTmu} i\}$ denote the jobs assigned to machine $i\in \mathcal{M}$ by $\LEPTmu$ and $n_i:=|J_i|$.  Without loss of generality we assume $n_i\geq 1$ for all $i\in\mathcal{M}$,
	as otherwise it is clear that $n<m$, and hence, $\LEPTmu$ is an optimal policy. Moreover, we define the expected workload of machine $i$ to be $$x_i:=\E[X_i^{\LEPTmu}]=\displaystyle{\sum_{j\in J_i}} \mathbb{E}[P_j].$$
\end{notation}
The next lemma gives bounds on the expected workload of any machine in an $\LEPTmu$ schedule. Interestingly, the gap between the lower
and upper bounds becomes smaller when the number of jobs scheduled on a machine grows.

\begin{lemma}\label{lemm:bound_ell}
 There exists $\ell > 0$ such that for all $i\in\mathcal{M}$,
\begin{equation*}
 \ell \leq x_i \leq \frac{n_i}{n_i-1} \ell,
\end{equation*}
 where we use the convention $\frac{n_i}{n_i-1} = \frac{1}{0}:=+\infty$ whenever $n_i=1$.
\end{lemma}

\begin{proof}
	We set $\ell:= \min\{x_i : i\in \mathcal{M}\}$, which is strictly positive by our assumption $n>m$. Then, the first inequality follows immediately. 
	Next, we will show that in each step that $\LEPTmu$ assigns a job to a machine the second inequality is fulfilled. Let $j$ denote the job which is put on machine $i$ in the current step. Furthermore, let $\ell'$ and $\ell$ denote the minimum expected load  among all machines before and after the allocation, respectively. Trivially, $\ell'\leq\ell$ is true. Moreover, let $x_i'$ and $x_i$ denote the expected workload of $i$ before and after assigning $j$ to it, respectively. Clearly, we have
	\begin{equation*}
	x_i=x_i'+\E[P_j].  
	\end{equation*}
	Observe, that $\ell' = x_i'$, because $\LEPTmu$ assigns $j$ to the machine with the smallest expected load. In addition, let $n_i$ denote the number of jobs running on machine $i$ after the insertion of $j$. Since $\LEPTmu$ sorts jobs in decreasing order of their expected processing times, it holds
	\begin{equation*}
	\E[P_j] \leq \frac{x_i'}{n_i-1} = \frac{\ell'}{n_i-1}.
	\end{equation*}
	Consider a machine other than $i$. If the inequality of the statement was fulfilled in an earlier step, then by setting the new $\ell$ it still is true. In the beginning, when we have no job at all, the inequality is true, so we only have to take care of machine~$i$.\\
	Finally, we obtain on machine $i$
	\begin{alignat*}{2}
	&\frac{x_i}{\ell} && = \frac{x_i'+\E[P_j]}{\ell} \leq \frac{x_i'+\E[P_j]}{\ell'} \leq 1 + \frac{\E[P_j]}{\ell'} \leq 1 + \frac{\ell'}{\ell' (n_i-1)} = \frac{n_i}{(n_i-1)}.
	\end{alignat*}
\end{proof}

We also need the following result on a specific convex optimization problem.
\begin{lemma}\label{lemm:boundCOP}
	Let $h:[0,1]\to\R$ be convex, $m\in \mathbb{Z}$ and $\sigma\in\mathbb{R}_{\geq 0}$. Moreover. let $v$ denote the optimal value of the following optimization problem
	\begin{subequations}
		\begin{alignat}{2}
		\underset{\vec{z} \geq 0}{\vec{\operatorname{maximize}}} \ &\quad \sum_{i=1}^m\  h(z_i) \tag{COP}\label{COP} \\
		s.t.\ &\quad \sum_{i=1}^m z_i= \sigma \nonumber\\
		&\quad  0\leq \vec{z} \leq 1.\nonumber
		\end{alignat}
	\end{subequations}
	Then, we have
	$$
	v\leq (m-\sigma)h(0) + \sigma h(1).
	$$
\end{lemma}
\begin{proof}
	Because $h$ is convex we have for all $z_i\in[0,1]$
	$$
	h(z_i)\leq h(0) + z_i(h(1)-h(0)).
	$$
	Then, if $\sum_{i=1}^m z_i=\sigma$, summing the above inequality over $i\in\{1,\ldots,m\}$ yields the desired result:
	$$
	\sum_{i=1}^{m} h(z_i)\leq m\cdot h(0) + \sigma(h(1)-h(0)).
	$$
\end{proof}
Moreover, we make use of the following technical convexity result.
\begin{lemma}\label{lemma:convex}
	The function $h:[0,1]\to\R $ with $y\mapsto (1-y)^{1+\frac{\ell}{y}}$, where $h(0)=e^{-\ell}$ using a continuous extension, is convex.
\end{lemma}

\begin{proof}
In order to show this, we compute its second derivative
\begin{equation*}
h''(y) = \frac{\ell (1-y)^{\frac{\ell}{y}-1}}{y^4} h_2(y),
\end{equation*}
where $h_2(y) := y^2 (\ell-y+2)+\ell (y-1)^2 \log ^2(1-y)-2 (\ell+1) (y-1) y \log (1-y).$
Now, we use the fact that $\log (1-y) = -\sum_{k=1}^\infty \frac{y^k}{k}$ for all $y\in [0,1)$. 
Hence, $\log^2 (1-y) = \sum_{k=2}^\infty \gamma_k y^k$, where $\gamma_k:=\sum_{i=1}^{k-1} \frac{1}{i(k-i)}$. After some calculus,
the terms of order $2$ and $3$ vanish and we obtain the following series representation of $h_2$ over $[0,1)$:
\begin{equation*}
h_2(y) = (\frac{\ell}{4}+\frac{1}{3}) y^4 + \sum_{k=5}^\infty  (\frac{2(\ell+1)}{(k-1)(k-2)} + \ell( \gamma_k + \gamma_{k-2} - 2 \gamma_{k-1})) y^k.
\end{equation*}

We are going to show that $\gamma_k + \gamma_{k-2} - 2 \gamma_{k-1}\geq 0$ for $k\geq 5$ implying that $h''(y)\geq 0$ for all $y\in [0,1)$. To do so, we rewrite the sums using the partial fraction decomposition. As a consequence, we obtain
\begin{alignat*}{2}
&\gamma_k + \gamma_{k-2} - 2 \gamma_{k-1} && = \frac{2}{k}\sum_{i=1}^{k-1} \frac{1}{i} + \frac{2}{k-2}\sum_{i=1}^{k-3} \frac{1}{i} - \frac{4}{k-1}\sum_{i=1}^{k-2} \frac{1}{i}\\
& && = \frac{2}{k}\left(\frac{1}{k-2} + \frac{1}{k-1}\right) - \frac{4}{(k-1)(k-2)} + \left(\frac{2}{k} + \frac{2}{k-2} - \frac{4}{k-1}\right)\sum_{i=1}^{k-3} \frac{1}{i} \\
& && = - \frac{6}{k(k-1)(k-2)} + \frac{4}{k(k-1)(k-2)}\sum_{i=1}^{k-3} \frac{1}{i} \\
& && \geq 0.
\end{alignat*}
The last inequality results from the fact that for all $k\geq 5$ we have $4\sum_{i=1}^{k-3}\frac{1}{i} \geq 6$.
Hence, $h$ is convex on $[0,1)$, and even on $[0,1]$ by continuity.
\end{proof}

Finally, we are ready to prove the main result of this section.
\begin{theorem}\label{theo:1+e-1}
 For short jobs only it holds
 \begin{equation*}
   \frac{\LEPTmu}{m\max(\rho,1)} \leq 1+e^{-1}.
 \end{equation*}
In particular, $\LEPTmu$ is an $(1+e^{-1})$-approximation algorithm in the class $\mathcal{P}$, over
the set of instances with short jobs only.
\end{theorem}

\begin{proof}
By Lemma~\ref{lemm:twopoints} we can bound the expected cost incurred on machine $i$ as
\begin{equation}
\E[\max(X_i^{\LEPTmu},1)]  \leq \sum_{j\in J_i} \mathbb{E}[P_j] + \prod_{j\in J_i}(1-\mathbb{E}[P_j])
 \leq x_i + \left(1-\frac{x_i}{n_i}\right)^{n_i} \label{bound-machinei}
\end{equation}
where the last inequality follows from the Schur-concavity of $\vec{\mu}\mapsto \prod_{j\in J_i} (1-\mu_j) $
over $[0,1]^{n_i}$;
cf.~\cite[Proposition 3.E.1]{MOA79}.
Next, we apply Lemma~\ref{lemm:bound_ell}, so there exists an $\ell> 0$ such that 
$\ell \leq x_i \leq \frac{n_i}{n_i-1} \ell$.
Let $y_i:=x_i-\ell\geq 0$.
The second inequality can be rewritten as
\begin{equation}
n_i \leq \frac{x_i}{x_i-\ell}=1+\frac{\ell}{y_i},\label{ni}
\end{equation}
which remains valid for $y_i=0$ if we define $\ell/0:=+\infty$.
We know that $P_j\in[0,1]$ almost surely, in particular $\mathbb{E}[P_j]\leq 1$, and hence, $x_i\leq n_i$. For this reason, the above inequality implies $x_i\leq \frac{x_i}{x_i-\ell}$
and therefore, $y_i\leq 1$.
By combining~\eqref{bound-machinei} and~\eqref{ni}, and using the fact that $(1-\frac{x_i}{n_i})^{n_i}$ is a nondecreasing function of $n_i$, we obtain
\begin{equation}
\E[\max(X_i^{\LEPTmu},1)] \leq x_i + \left(1-(x_i-\ell)\right)^{\frac{x_i}{x_i-\ell}} = \ell + y_i + h(y_i), \label{bound2-machinei}
\end{equation}
where $h$ is defined as in Lemma~\ref{lemma:convex}. Note that
the $y_i$'s satisfy $y_i\in[0,1]$. 
Summing up the inequalities~\eqref{bound2-machinei} over all $i\in\mathcal{M}$ and using the fact that $\sum_{i\in \mathcal{M}} y_i = m(\rho-\ell)$
we have
\begin{equation*}
\LEPTmu = \sum_{i\in \mathcal{M}} \E[\max(X_i^{\LEPTmu},1)] \leq \rho m + \sum_{i\in \mathcal{M}} h(y_i). 
\end{equation*}
By convexity of $h$ due to Lemma~\ref{lemma:convex} we can use Lemma~\ref{lemm:boundCOP} by setting $\sigma:=m(\rho-\ell)$ yielding 
$$
\rho m + \sum_{i\in \mathcal{M}} h(y_i) \leq m\rho + m(1-\rho+\ell)e^{-\ell} \leq m(\rho+e^{-\rho}),
$$
where the last inequality follows from the fact that $(1+\ell-\rho)e^{-\ell}$ is a nondecreasing function of $\ell$ over $[0,\rho]$ as $0<\ell\leq\rho$.
As a consequence, we obtain
\begin{equation}
\frac{\LEPTmu}{m\max(\rho,1)} \leq \frac{\rho + e^{-\rho}}{\max(\rho,1)}. 
\label{betterbound}
\end{equation}
The theorem now follows since the above ratio is maximized for $\rho=1$.
This is true because 
$\frac{\rho + e^{-\rho}}{\max(\rho,1)} = \rho + e^{-\rho}$ on~$[0,1]$, hence increasing, 
and 
$\frac{\rho + e^{-\rho}}{\max(\rho,1)} = 1 + \frac{e^{-\rho}}{\rho}$ on~$[1,+\infty)$, hence decreasing.\\
Combining this result with the inequality $OPT_\mathcal{P}\geq m \max(1,\rho)$ from Proposition~\ref{prop:order} yields the approximation ratio of $\LEPTmu$.
\end{proof}
\begin{remark}
	For a specific $\rho$, the right hand side of (\ref{betterbound}) gives tighter bounds than $1+e^{-1}$.
\end{remark}

\section{The general case: Taking long jobs into account}\label{sec:long}
In this section, we show that $\LEPTmu$ has performance guarantee $(1+e^{-1})$ even for instances containing long jobs, i.e.,
jobs whose duration may exceed $1$ with positive probability. It can be shown --using a similar approach as in Theorem~\ref{theo:1+e-1}--
that $\frac{\LEPTmu}{OPT_{\mathcal{R}}}\leq1+e^{-\frac{1}{d_{\max}}}$ for instances where each job satisfies 
$P_j\in[0,d_{\max}]$ almost surely for some $d_{\max}\geq 1$, and that this bound is tight.
Letting $d_{\max}\to\infty$ just gives the trivial approximation guarantee of $2$, so we have to use a better lower bound
on $OPT_\mathcal{P}$ in order to prove that $\LEPTmu$ is a $(1+e^{-1})$-approximation algorithm. 
Our next candidate as
the lower bound on $OPT_{\mathcal{P}}$ is $\E[OPT(\vec{P})]$, cf.\ Proposition~\ref{prop:order}. Let us first introduce some further notation.
\begin{notation}
We denote the sum of expected processing times by $s:=\sum_{j\in\mathcal{J}} \E[P_j].$
We split each job into a truncated
part $P_j':=\min(P_j,1)$ and an excess
part $P_j'':=\max(P_j-1,0)$, so that $P_j=P_j'+P_j''$
and $P_j'\leq 1$. We further define 
the truncated load of machine $i$ according to $\LEPTmu$ by $\alpha_i:=\sum_{j\in J_i} \E[P_j']$ and the excess of machine $i$ by $\beta_i:= \sum_{j\in J_i} \E[P_j'']$. 
Moreover, $\beta$ is the overall excess, i.e., $\beta := \sum_{j\in\mathcal{J}} \E[P_j'']=\sum_{i\in \mathcal{M}}\beta_i$.
\end{notation}

Using this notation, one can easily observe that there exists $\ell> 0$ such that for all $i\in\mathcal{M}$
\begin{equation}
\ell \leq \alpha_i + \beta_i \leq \frac{n_i}{n_i-1}\ell, \quad \sum_{i\in\mathcal{M}} \alpha_i = s-\beta \quad \text{and }\quad \sum_{i\in\mathcal{M}} \beta_i = \beta,
\label{observation}
\end{equation}
where the first statement immediately follows by Lemma~\ref{lemm:bound_ell}.
First, we require a lemma that relates an instance to its truncated version with respect to the value of an optimal anticipative policy. 

\begin{lemma}\label{lemm:reduction}
		The truncated jobs with processing times $\vec{P}'$ are short and we have
		\begin{equation*}
		\E[OPT(\vec{P}')]=\E[OPT(\vec{P})]-\beta.
		\end{equation*} 
\end{lemma}

\begin{proof}
	Let $\vec{p}$ be an arbitrary but fixed realization of the processing times for instance $I$, and
	let $\vec{p}'$ denote the corresponding truncated processing times, i.e., $p_j'=\min(p_j,1)$. Furthermore, let $\Pi$ be an arbitrary policy and $\Pi(\vec{p})$ be the assignment resulting from $\Pi$ for realization $\vec{p}$. Moreover, let $J_i(\vec{p})$ denote the set of jobs assigned to machine $i$ by $\Pi(\vec{p})$.
	The difference between the costs incurred by $\Pi(\vec{p})$ for the realizations $\vec{p}$ and $\vec{p}'$ on machine $i$ is
	\begin{equation*}\label{diff_omega}
	\phi(\Pi(\vec{p}),\vec{p})-\phi(\Pi(\vec{p}),\vec{p}')
	= \max( \sum_{j\in J_i(\vec{p})} p_j ,1) - \max( \sum_{j\in J_i(\vec{p})} \min(p_j,1) ,1). 
	\end{equation*}
	One can show that $\phi(\Pi(\vec{p}),\vec{p})-\phi(\Pi(\vec{p}),\vec{p}') = \sum_{j\in J_i(\vec{p})} \max(p_j-1,0).$
	Inserting $OPT(\vec{p}')$ and $OPT(\vec{p})$ for $\Pi$ and taking the expectation we obtain
	\begin{equation}\label{OPT_P'}
	\E[\phi(OPT(\vec{P}'),\vec{P})] = \E[\phi(OPT(\vec{P}'),\vec{P}')] + \beta
	\end{equation}
	and
	\begin{equation}\label{OPT_P}
	\E[\phi(OPT(\vec{P}),\vec{P}')] = \E[\phi(OPT(\vec{P}),\vec{P})] - \beta.
	\end{equation}
	Using optimality of $OPT(\vec{p}')$ for realization $\vec{p}'$ and \eqref{OPT_P'} we have 
	\begin{equation}
	\E[\phi(OPT(\vec{P}'),\vec{P}')]\leq \E[\phi(OPT(\vec{P}),\vec{P}')] = \E[\phi(OPT(\vec{P}),\vec{P})] - \beta. \label{ineqFA1} 
	\end{equation}
	Similarly, using optimality of $OPT(\vec{p})$ for realization $\vec{p}$ and \eqref{OPT_P} we obtain
	\begin{equation}
	\E[\phi(OPT(\vec{P}),\vec{P})]\leq \E[\phi(OPT(\vec{P}'),\vec{P})] = \E[\phi(OPT(\vec{P}'),\vec{P}')] + \beta. \label{ineqFA2} 
	\end{equation}
	Finally, by combining~\eqref{ineqFA1} and~\eqref{ineqFA2} we have
	$$
	\E[\phi(OPT(\vec{P}),\vec{P})] -\beta \leq \E[\phi(OPT(\vec{P}'),\vec{P}')]  \leq \E[\phi(OPT(\vec{P}),\vec{P})] - \beta,
	$$
	concluding the proof.
\end{proof}

Next, we use this Lemma to obtain a handy lower bound on $\E[OPT(\vec{P})]$.	
\begin{lemma}\label{lemm:LB}
	$\E[OPT(\vec{P})]\geq \max(s,m+\beta)$.
\end{lemma}
\begin{proof}
	We claim that $\E[OPT(\vec{P})]=\E[OPT(\vec{P}')]+\beta$.
	Then, the result follows from 
	$\E[OPT(\vec{P}')]\geq \max(\sum_{j\in\mathcal{J}} \E[P_j'],m)$ using Proposition~\ref{prop:order},
	and the identity $\sum_{j\in\mathcal{J}} \E[P_j'] = s - \beta$,
	implying that
	$$\E[OPT(\vec{P})]\geq \max(s-\beta,m)+\beta = \max(s,m+\beta).$$
\end{proof}

We continue with several lemmas obtaining upper bounds on $\LEPTmu$ sequentially. 
\begin{lemma}\label{lemm:v1}
	The cost of $\LEPTmu$ can be bounded from above by 
	$$\LEPTmu\leq s+\sum_{i\in\mathcal{M}} \left(1-\frac{\alpha_i}{n_i}\right)^{n_i}.$$
\end{lemma}
\begin{proof}
	The cost on machine~$i$ for realization $\vec{p}$ is
	\begin{equation*}
	\max(\sum_{j\in J_i} p_j,1) = \max(\sum_{j\in J_i} p_j' ,1) + \sum_{j\in J_i} \max(p_j-1,0),
	\end{equation*}
	where $p_j':=\min(p_j,1)$. This follows 
	by distinguishing between the cases
	where 
	\mbox{$\forall j\in J_i:\ p_j\leq 1$} or \mbox{$\exists j\in J_i: p_j>1$}.
	Summing up this equality over all machines and taking the expectation yields the identity
	$$
	\LEPTmu = \sum_{i\in\mathcal{M}} \E[\max(\sum_{j\in J_i} P_j',1)] + \beta.
	$$
	Since the reduced jobs $P_j'$ are short, we can  
	use Lemma~\ref{lemm:twopoints} to obtain  
	\begin{equation*}
	\LEPTmu \leq \sum_{j\in\mathcal{J}} \E[P_j'] + \sum_{i\in\mathcal{M}} \prod_{j\in J_i} \left(1- \E[P_j']\right)+\beta
	\leq s + \sum_{i\in\mathcal{M}} \left(1-\frac{\alpha_i}{n_i}\right)^{n_i},
	\end{equation*}
	where the second inequality follows from the identity $\sum_{j\in\mathcal{J}} \E[P_j'] = s - \beta$
	and the Schur-concavity of $\vec{\mu}\mapsto \prod_{j\in J_i} (1-\mu_j)$ similarly as in the proof of Theorem~\ref{theo:1+e-1}.
\end{proof}

Next, we will modify our arbitrary instance, which will only worsen the approximation ratio of $\LEPTmu$. To do so, we introduce a partition of the machines.
\begin{notation}\label{notation:machinepartition}
	We partition the machines into the following three types:
	\begin{align*}
	&\mathcal{M}_0:=\{i\in\mathcal{M}: \alpha_i>2\ell\}\\
	&\mathcal{M}_1:=\{i\in\mathcal{M}: \ell \leq \alpha_i\leq 2\ell\}\\
	&\mathcal{M}_2:=\{i\in\mathcal{M}: \alpha_i<\ell\}.
	\end{align*}
	Furthermore, let $m_k:=\vert\mathcal{M}_k\vert$ for $k\in\{0,1,2\}$.
\end{notation}

\begin{figure}
	\centering
	\begin{minipage}{.5\textwidth}
		\centering
		\includegraphics[scale=0.7]{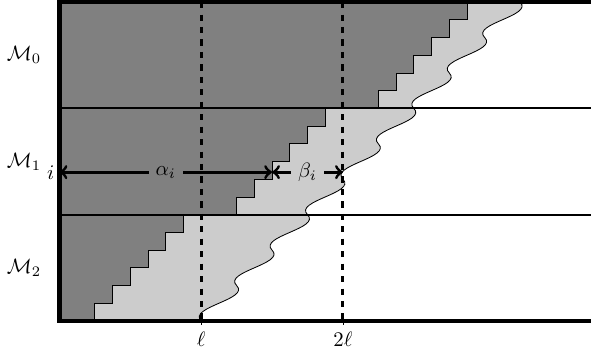}
	\end{minipage}%
	\begin{minipage}{.5\textwidth}
		\centering
		\includegraphics[scale=0.7]{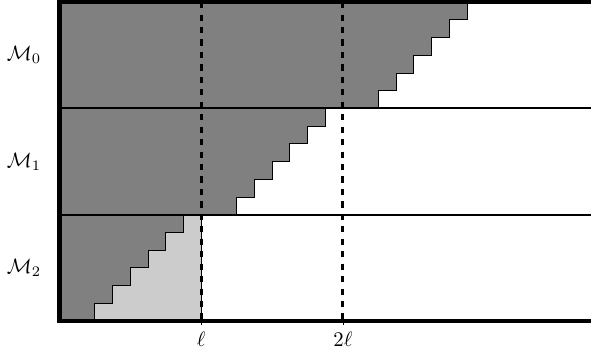}
	\end{minipage}
	\caption{\small Illustration of the machine partition in Notation~\ref{notation:machinepartition} on the left, with expected load $x_i=\alpha_i+\beta_i\geq \ell$
 on the horizontal axis. The dark segments indicate $\alpha_i$, the light gray segments indicate $\beta_i$. On the right the corresponding modification of Lemma~\ref{lemm:v2} from $\beta_i$ to $\hat\beta_i$. }
	\label{fig:machinepartition}
\end{figure}
\begin{lemma}\label{lemm:v2}
	Let
	\begin{alignat*}{2}
	&\hat\beta_i:=\begin{cases}
		0,  & \text{if }i\in \mathcal{M}_0\cup\mathcal{M}_1\\
		\ell - \alpha_i, & \text{if }i\in \mathcal{M}_2
	\end{cases}&& \qquad \hat\beta:=\sum_{i\in\mathcal{M}}\hat{\beta_i}
	\end{alignat*}
	\begin{alignat*}{3}
	&\tilde s:=s+\hat\beta-\beta &&\qquad  \hat s:=\tilde s -\sum_{i\in \mathcal{M}_0}\alpha_i &&\qquad \hat m:=m-m_0.
	\end{alignat*}
	Then we have for all $i\in\mathcal{M}$
	$$
	\sum_{i\in\mathcal{M}} \alpha_i = \hat s-\hat\beta \quad \text{and } \ell \leq \alpha_i + \hat\beta_i \leq \frac{n_i}{n_i-1}\ell
	$$
	and
	$$
	\frac{\LEPTmu}{\max(s,m+\beta)}\leq \frac{\hat s +  \sum_{i\in\mathcal{M}_1\cup\mathcal{M}_2}\  \left( 1- \frac{\alpha_i}{n_i} \right)^{n_i}}{\max(\hat s,\hat m+\hat\beta)}.
	$$ 
\end{lemma}
\begin{proof}
	Observe that $\beta \geq\hat\beta\geq 0$ and hence, $\tilde s \leq s$. By definition it holds $\sum_{i\in\mathcal{M}} \alpha_i = s-\beta =\tilde s-\hat\beta$. Moreover, we have $\ell\leq \alpha_i+\hat \beta_i\leq \frac{n_i}{n_i-1}\ell$ using the identity
	$$
	\alpha_i+\hat\beta_i=\begin{cases}
	\alpha_i  & \forall i\in\mathcal{M}_0\cup\mathcal{M}_1\\
	\ell & \forall i\in\mathcal{M}_2.
	\end{cases}
	$$
	Therefore, using Lemma~\ref{lemm:v1} we obtain
	\begin{alignat*}{2}
	&\frac{\LEPTmu}{\max(s,m+\beta)}&& \leq \frac{ s+\sum_{i\in\mathcal{M}}\  \left( 1- \frac{\alpha_i}{n_i} \right)^{n_i}}{\max( s,m+\beta)}\\
	& && = \frac{\tilde s+\beta-\hat\beta+\sum_{i\in\mathcal{M}}\  \left( 1- \frac{\alpha_i}{n_i} \right)^{n_i}}{\max(\tilde s,m+\hat\beta)+\beta-\hat\beta }\\
	& && \leq \frac{\tilde s+\sum_{i\in\mathcal{M}}\  \left( 1- \frac{\alpha_i}{n_i} \right)^{n_i}}{\max(\tilde s,m+\hat\beta) }.
	\end{alignat*}
	Next, we want to show that we can further reduce to an instance in which we do not have any machine of type $\mathcal{M}_0$. Assume that $ i_0\in\mathcal{M}_0\neq\emptyset$. Since we consider a truncated instance, we know by Lemma~\ref{lemm:bound_ell} that $n_{i_0}=1$ and as mentioned in Theorem~\ref{theo:1+e-1} that we also have $\alpha_{i_0}\leq n_{i_0}$. This yields
	$$
	2\ell< \alpha_{i_0}\leq n_{i_0}=1.
	$$
	Hence, $\ell$ is bounded from above by $\frac{1}{2}$.  Moreover, we know that $\tilde s \leq m+\tilde\beta$, because it holds
	\begin{equation*}
	 \tilde s =\sum_{i\in \mathcal{M}} \alpha_i+\hat \beta_i
	 \leq m_0 +2\ell m_1 + \ell m_2
	 < m_0 + m_1 + \frac{1}{2}m_2+\hat\beta
	\leq m \leq m +\hat \beta.
	\end{equation*}
	Observe that $\hat m>0$ since $\mathcal{M}_1\neq\emptyset$ as we can set $\ell:=\min_{i\in\mathcal{M}} \alpha_i+\beta_i$. Combining all the results above, we obtain
	\begin{alignat*}{2}
	& \frac{\tilde s+\sum_{i\in\mathcal{M}}\  \left( 1- \frac{\alpha_i}{n_i} \right)^{n_i}}{\max(\tilde s,m+\hat\beta) }&&=\frac{\tilde s+\sum_{i\in\mathcal{M}_0}\  \left( 1- \alpha_i \right) + \sum_{i\in\mathcal{M}_1\cup\mathcal{M}_2}\  \left( 1- \frac{\alpha_i}{n_i} \right)^{n_i}}{\max(\tilde s,m+\hat\beta) }\\
	& &&=\frac{\hat s + \sum_{i\in \mathcal{M}_0}\alpha_i + \sum_{i\in\mathcal{M}_0}\  \left( 1- \alpha_i \right) + \sum_{i\in\mathcal{M}_1\cup\mathcal{M}_2}\  \left( 1- \frac{\alpha_i}{n_i} \right)^{n_i}}{m+\hat \beta}\\
	& &&=\frac{\hat s + m_0+ \sum_{i\in\mathcal{M}_1\cup\mathcal{M}_2}\  \left( 1- \frac{\alpha_i}{n_i} \right)^{n_i}}{\hat m+ m_0+\hat \beta}\\
	& &&\leq\frac{\hat s +  \sum_{i\in\mathcal{M}_1\cup\mathcal{M}_2}\  \left( 1- \frac{\alpha_i}{n_i} \right)^{n_i}}{\hat m+ \hat \beta}\\
	& &&=\frac{\hat s +  \sum_{i\in\mathcal{M}_1\cup\mathcal{M}_2}\  \left( 1- \frac{\alpha_i}{n_i} \right)^{n_i}}{\max(\hat s,\hat m+\hat\beta)}
	\end{alignat*}
	where we used in the last step that $\hat s\leq \hat m+\hat\beta$, since we have
	$$
	\hat s = \tilde s - \sum_{i\in \mathcal{M}_0}\alpha_i=\sum_{i\in\mathcal{M}_1\cup\mathcal{M}_2}\alpha_i+\hat \beta_i\leq2\ell m_1 + \ell m_2 < m_1 + \frac{1}{2}m_2\leq \hat m \leq \hat m +\hat \beta.
	$$
\end{proof}
We now simplify the sum we just derived to obtain a more structured upper bound.

\begin{lemma}\label{lemm:v3}
	Let $v_1^*$ and $v_2^*$ denote the optimal value of the following convex optimization problems, respectively,
	{\small
	\begin{subequations}
		\begin{alignat}{4}
		\underset{\vec{\alpha}\geq 0}{\vec{\operatorname{maximize}}} \ &\quad \sum_{i\in\mathcal{M}_1}  \left( 1- (\alpha_i-\ell) \right)^{\frac{\alpha_i}{\alpha_i-\ell}} &&\qquad \quad \qquad\qquad\quad \underset{\vec{\alpha}\geq 0}{\vec{\operatorname{maximize}}} && \quad \sum_{i\in\mathcal{M}_2} e^{-\alpha_i} &&\tag{OP}\label{OP} \\
		s.t.\ &\quad \sum_{i\in\mathcal{M}_1} \alpha_i =\hat s-m_2\ell &&
		\hspace{3.8cm}s.t.\  && \quad \sum_{i\in\mathcal{M}_2} \alpha_i = m_2\ell - \hat\beta &&\nonumber\\
		&\quad \ell\leq \alpha_i\leq \ell+1 &&\forall i\in\mathcal{M}_1.  &&  \quad 0\leq \alpha_i \leq \ell &&\forall i\in\mathcal{M}_2.  \nonumber
		\end{alignat}
	\end{subequations}
	}
	Then, we have
	$$
	\sum_{i\in\mathcal{M}_1\cup\mathcal{M}_2}\  \left( 1- \frac{\alpha_i}{n_i} \right)^{n_i}\leq v_1^*+v_2^*.
	$$
\end{lemma}
\begin{proof}
	By Lemma~\ref{lemm:v2} we can rewrite
	$$
	\sum_{i\in\mathcal{M}} \hat\beta_i = \hat\beta\quad \text{and} \quad \sum_{i\in\mathcal{M}} \alpha_i =  \hat s-\hat\beta
	$$
	as
	$$
	\sum_{i\in\mathcal{M}_2}\alpha_i=m_2\ell-\hat\beta\quad \text{and} \quad \sum_{i\in\mathcal{M}_1}\alpha_i=\hat s-m_2\ell.
	$$
	We immediately obtain the first constraints of both optimization problems~\eqref{OP}.
	It remains to show the constraint $\alpha_i\leq \ell+1$ for all $i\in\mathcal{M}_1$, as the others follow immediately from Lemma~\ref{lemm:v2}. Using the same ideas as in the proof of Theorem~\ref{theo:1+e-1} we have
	$$
	\alpha_i\leq n_i\leq \frac{\alpha_i}{\alpha_i-\ell}
	$$
	implying our desired constraint. The above inequality also yields the upper bound on the summands over $\mathcal{M}_1$ as in the same Theorem. We can bound the summands of the other sum more coarsely from above by the exponential function. Obviously, the optimal values of \eqref{OP} with variables $\vec{\alpha}$ can only increase the sum.
\end{proof}

Finally, we are ready to state the approximation ratio of $\LEPTmu$.
\begin{theorem}\label{GeneralApproximationTheorem}
	$\LEPTmu$ is a $(1+e^{-1})$-approximation algorithm.
\end{theorem}
\begin{proof}
	Combining the Lemmas~\ref{lemm:LB} to \ref{lemm:v3}, we obtain
	$$
	\frac{\LEPTmu}{OPT_{\mathcal{P}}}\leq \frac{\hat s+v_1^*+v_2^*}{\max(\hat s, \hat m+\hat \beta)},
	$$
	where we keep the notation of Lemma~\ref{lemm:v3}. Applying Lemma~\ref{lemm:boundCOP} to both sums by scaling and translating $\vec{\alpha}$, such that the variables are in the unit interval, yields
	\begin{alignat*}{2}
	&\frac{\hat s+v_1^*+v_2^*}{\max(\hat s,\hat m+\hat\beta)} && \leq \frac{\hat s + \left[(\hat s-\hat m\ell)\cdot0 + (\hat m-k-(\hat s-\hat m\ell))e^{-\ell}\right] + \left[\left(k-\frac{\hat\beta}{\ell}\right)e^{-\ell}+\frac{\hat\beta}{\ell}\cdot1\right]}{\max(\hat s,\hat m+\hat\beta)}\\
	& && = \frac{\hat s\left(1-e^{-\ell}\right)}{\max(\hat s,\hat m+\hat\beta)} + \frac{\left(\hat m+\hat m\ell-\frac{\hat\beta}{\ell}\right)e^{-\ell} +\frac{\hat\beta}{\ell}}{\max(\hat s,\hat m+\hat\beta)}\\
	& && \leq 1-e^{-\ell} + \frac{\left(\hat m+\hat m\ell-\frac{\hat\beta}{\ell}\right)e^{-\ell} +\frac{\hat\beta}{\ell}}{\hat m+\hat\beta}.
	\end{alignat*}
	Using the variable transformation $u=\frac{\hat\beta}{\hat m}\geq0$ for the fractional term we obtain
	\begin{equation*}
	 1-e^{-\ell} + \frac{\left(1+\ell-\frac{u}{\ell}\right)e^{-\ell} +\frac{u}{\ell}}{1+u} =1 + e^{-\ell}\cdot\frac{\frac{u}{\ell}e^{\ell}-u -\frac{u}{\ell}+\ell}{1+u},
	\end{equation*}
	where we define the right hand side as $f_\ell:\mathbb{R}_{\geq 0}\to\mathbb{R}$ for $\ell>0$. Straightforward calculation shows that its derivative is
	$$
	\frac{d}{du}f_\ell(u)=\underbrace{\frac{e^{-\ell}}{\ell(1+u)^2}}_{>0}\cdot\left(e^\ell -\ell^2-\ell-1\right).
	$$
	Hence,
	$$
	f_\ell\text{ is}\begin{cases}
	\text{nonincreasing, }  & \text{if }\ell\in(0,r]\\
	\text{nondecreasing, } & \text{if }\ell\in[r,\infty),
	\end{cases}
	$$
	where $r$ denotes the (strictly positive) root of the term in the brackets, i.e., $e^{r} -r^{2}-r-1=0$.
	As a consequence, for $\ell\in(0,r]$ we obtain for all $u\geq0$
	$$
	f_\ell(u)\leq f_\ell(0)=1 + \ell e^{-\ell}\leq 1+ e^{-1},
	$$
	where the last inequality holds as the map $\ell\mapsto 1+\ell e^{-\ell}$ attains its maximum at $\ell=1$ using basic calculus. On the other hand, for $\ell\in[r,\infty)$ we have for all $u\geq 0$
	$$
	f_{\ell}(u)\leq\lim_{u'\to \infty} f_{\ell}(u')=1+\frac{1}{\ell}-e^{-\ell}-\frac{e^{-\ell}}{\ell}\leq 1+\frac{1}{r}-e^{-r}-\frac{e^{-r}}{r}.
	$$
	The last inequality results uses the fact that the map $\ell\mapsto 1+\frac{1}{\ell}-e^{-\ell}-\frac{e^{-\ell}}{\ell}$ attains its maximum at $r$. We can further simplify this to
	\begin{equation*}
	1+\frac{1}{r}-e^{-r}-\frac{e^{-r}}{r}+0= 1+\frac{1}{r}-e^{-r}-\frac{e^{-r}}{r}-\frac{1}{r}+re^{-r}+e^{-r}+\frac{e^{-r}}{r}= 1+re^{-r}\leq 1 + e^{-1},
	\end{equation*}
	where the first equality follows by the definition of $r$ and the last inequality using the same argument as in the other case. This concludes the proof of the theorem.
\end{proof}

\section{The Price of Fixed Assignments}\label{sec:pofa}

In this section, we are going to show that the price of fixed assignments is equal to $1+e^{-1}$. 

\begin{theorem}\label{theorem:PoFA}
 The price of fixed assignments for SEBP is equal to $(1+e^{-1})$:
 \begin{equation*}
  \qquad\qquad\qquad\PoFA = 1+e^{-1}.
 \end{equation*}
\end{theorem}
\begin{proof}
Let $I=(\vec{P},m)$ denote an instance of SEBP.
Proposition~\ref{prop:order} and Theorem~\ref{GeneralApproximationTheorem} yields
 \begin{equation*}
 \frac{OPT_{\mathcal{F}}(I)}{OPT_{\mathcal{P}}(I)} \leq \frac{\LEPTmu(I)}{\E[OPT(\vec{P})]}
 \leq 1+e^{-1}.
 \end{equation*}
Therefore, it remains to show that for all $\epsilon>0$ there exists an instance $I$ in which we have
\begin{equation*}
\frac{OPT_{\mathcal{F}}(I)}{OPT_{\mathcal{P}}(I)} \geq 1+e^{-1} - \epsilon. 
\end{equation*}
For this purpose, we consider an instance $I=(\vec{P},m)$ in which we have $n=km$ jobs for some $k\in\mathbb{N}$,
where $P_j\sim \operatorname{Bernoulli}(\frac{1}{k})$ for all $j\in\mathcal{J}$.
An optimal fixed assignment policy assigns each machine the same number of jobs, in this case $k$.
The cost on one machine is hence the expected value of $\max\left(Z,1\right)$,
where $Z:=\sum_{j=1}^{k} P_{j}\sim\operatorname{Binomial}(k,\frac{1}{k})$. So,
\begin{align*}
 OPT_{\mathcal{F}}(I) = m \cdot \E[\max(Z,1)] &=m\cdot \Big( \E\big[Z| Z\geq 1\big]\ \mathbb{P}[Z\geq 1] + \E\big[1| Z< 1\big]\ \mathbb{P}[Z<1] \Big)\\ 
 & =m\cdot \big( \E[Z] + \mathbb{P}[Z=0]\big)
  = m\cdot \left( 1+ \left(1-\frac{1}{k}\right)^k \right),
\end{align*}
which converges to $m(1+e^{-1})$ as $k\to\infty$.
On the other hand, an optimal policy in $\mathcal{P}$ lets a job run whenever a machine becomes idle. 
The cost of an optimal policy is hence $m$ whenever less than $m$ jobs have duration $1$,
and is equal to $\sum_{j=1}^{km} p_j$ otherwise. This shows that
$OPT_{\mathcal{P}}(I) = \E[\max(U,m)]$, where $U:= \sum_{j=1}^{km} P_j \sim \operatorname{Binomial}\left(km,\frac{1}{k}\right)$.
Now, we can argue as in Theorem~\ref{theo:1+e-1} that $U$ converges in distribution to $Y\sim\operatorname{Poisson}(m)$ as $k\to\infty$.
So, by Lemma~\ref{lemm:poisson-max}, we have
\begin{equation*}
OPT_{\mathcal{P}}(I) \to m \cdot \big(1 + \frac{e^{-m} m^m}{m!}\big)\quad \text{as}\quad k \to \infty.
\end{equation*}

Finally, we have shown that the ratio of $OPT_{\mathcal{F}}(I)$ to 
$OPT_{\mathcal{P}}(I)$ can be made arbitrarily close to $(1+e^{-1})\cdot \big(1 + \frac{e^{-m} m^m}{m!}\big)^{-1}$
by choosing $k$ large enough. We conclude by observing that $\lim_{m\to \infty} \frac{m^me^{-m}}{m!}=0$,
so this ratio can be arbitrarily close to $1+e^{-1}$.
\end{proof}
This proves that our analysis of $LEPT_\mathcal{F}$ is tight. It even shows that $LEPT_\mathcal{F}$ is
the best fixed assignment policy in the following sense: Since there exists instances for which the
ratio of an optimal fixed assignment policy to an optimal non-anticipatory policy is 
arbitrarily close to $1+e^{-1}$ and the fact that $LEPT_\mathcal{F}$ is a $1+e^{-1}$-approximation,
we cannot hope to find a policy $\Pi\in\mathcal{F}$ with a
better approximation guarantee in the class $\mathcal{P}$.

\section{Performance of LEPT in the class of fixed assignment policies}\label{sec:ratioF}

It would also be interesting to characterize the approximation guarantee of $\LEPTmu$ in the class of fixed assignment policies.
The next proposition gives a lower bound:
\begin{proposition}
For all $\epsilon>0$,
there exists an instance $I$ of SEBP such that $\frac{\LEPTmu}{OPT_{\mathcal{F}}(I)}=\frac{4-\epsilon}{3}$.
\end{proposition}
\begin{proof}
 We construct an instance with $m=2$ machines and $n=3$ jobs. The first two jobs are deterministic and have duration $P_1=P_2=1$.
 The distribution of the third job is $P_3=\frac{1}{\epsilon} X$, where $X\sim\operatorname{Bernoulli}(\epsilon)$, so $\E[P_3]=1$.
 We assume that the $\LEPTmu$ policy assigns both deterministic jobs to the first machine and the stochastic job
 to the other machine, which gives $\LEPTmu=2+(1-\epsilon)+\frac{\epsilon}{\epsilon}=4-\epsilon$.
 In contrast, for any policy $\Pi^*$ which assigns the two deterministic jobs on different machines,
 we have $\Pi^*=1+(1-\epsilon) + (1+\frac{1}{\epsilon})\epsilon = 3$.
 The policy $\Pi^*$ reaches the lower bound $m\max(\rho,1)$ of Proposition~\ref{prop:order}, 
 hence it is optimal.
\end{proof}

This shows that the best approximation ratio
for $\LEPTmu$ in the class of fixed assignment policies lies between $\frac{4}{3}\approx 1.333$ and $1+e^{-1}\approx 1.368$.

\section{Restriction to a family of processing time distributions}\label{sec:stochastic_dominance}

The results of this section heavily rely on the notion
of second-order stochastic dominance, which was
introduced in the late 60's
to model the preferences of decision-makers 
regarding different gambles. We first give a short introduction with the necessary background on this topic.

\begin{definition}
	Let $Y$ and $Z$ be random variables with finite expectation. We say that $Y$ has second-order stochastic dominance over $Z$, and we write $Y\succeq_{(2)}Z$ if and only if
	$$
\int_{-\infty}^x (F_Z(t)-F_Y(t))\ dt \geq 0,\quad \forall x\in\R,
$$
with $F_Y$ and $F_Z$ the cumulative distribution functions of $Y$ and $Z$, respectively.
\end{definition}

Using the well-known fact that $\E[Y]-\E[Z] = \int_{-\infty}^{\infty} \big(F_Z(t)-F_Y(t)\big) dt$, it is easy to see
that a simple
sufficient condition for $Y\succeq_{(2)} Z$ is
that $\E[Y]\geq\E[Z]$ and that $F_Y$ and $F_Z$ are \emph{single-crossing}, i.e., 
for some $x_0$ we have $F_Y(t)-F_Z(t)\leq 0$
on the interval $(-\infty,x_0)$ and 
$F_Y(t)-F_Z(t)\geq 0$ over $[x_0,\infty)$:
\begin{equation}\label{singlecrossing}
\Big(\E[Y]\geq \E[Z] \ \text{ and }\ \exists x_0\in\R: \forall t\in\R,\big(F_Y(t)-F_Z(t)\big) (t-x_0)\geq 0 \Big)
\implies Y\succeq_{(2)} Z.
\end{equation}

In this work, we mostly use the
$\preceq_{(2)}$-ordering to compare random variables with the same mean. In this case, we shall see that $\preceq_{(2)}$ is linked to another dominance relation
relying on the concept of Lorenz curves.
\begin{definition}
	Let $Y$ be a random variable with finite expectation. The Lorenz curve of $Y$ is a function $L_Y:[0,1]\to [0,1]$ with 
	$$
	L_Y(p):= \frac{1}{\E[Y]} \int_0^p F_Y^{-1}(z) dz,
	$$
	where $F_Y^{-1}(z) = \inf \{y: F_Y(y)\geq z\}$ is the quantile function of $Y$.
\end{definition}
The Lorenz curve $L_Y$ of a random variable $Y$
is convex and nondecreasing, and it satisfies $L_Y(0)=0$, $L_Y(1)=1$. It was introduced by Lorenz~\cite{Lor1905}
to compare the distribution of income across
different countries: for a population with continuous distribution of income $F_Y$, $L_Y(p)$
represents the percentage of the total wealth owned by the bottom $100p$\% of all individuals. 
The situation where all individuals own the same wealth corresponds
to a deterministic variable ($Y=a$ for some $a\geq 0$) with Lorenz curve $L_Y(p)=p$, called the \emph{line of perfect equality}.
Based on the Lorenz curve we can define another dominance relation.
\begin{definition}
	Let $Y$ and $Z$ be nonnegative random variables with finite expectation. 
	We say that $Y$ Lorenz dominates $Z$,
	and we write $Y\succeq_L Z$, if and only if for all $p\in[0,1]$,
	$$
	L_Y(p)\leq L_Z(p).
	$$
\end{definition}
\begin{figure}[t]
	\scalebox{0.9}{\hspace{4.1cm}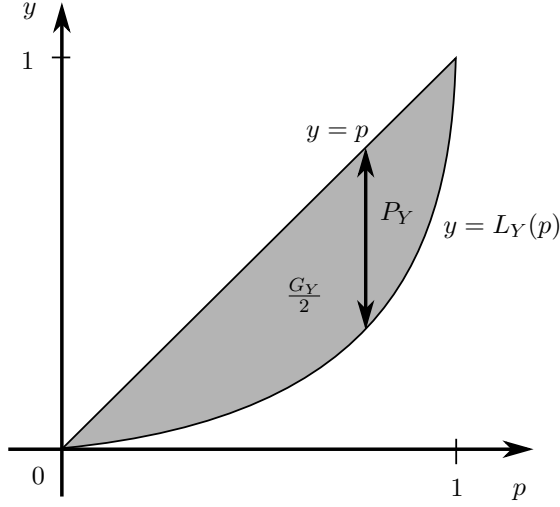}
	\caption{\small Lorenz curve of a nonnegative random variable $Y$. The Gini index $G_Y\in[0,1]$ corresponds to twice
		the shaded area, and the Pietra index $P_Y$ corresponds to the maximal vertical distance between the
		Lorenz curve and the line of perfect equality. \label{fig:lorenz}}
\end{figure}

We next summarize equivalent characterizations of $\preceq_{(2)}$ and $\preceq_L$ under the assumption that $Y$ and $Z$ have equal means:
\begin{proposition} \label{prop:equiv_stochdom}
 Let $Y$ and $Z$ be nonnegative random variables with $\E[Y]=\E[Z]<\infty$. The following statements are equivalent:
 
 \begin{enumerate}[label=\emph{(\roman*)},ref=(\roman*)]
  \item \label{stoch_l1}$Y\succeq_{(2)} Z$ 
  \item \label{stoch_l4}$Y \preceq_L\ Z$
  \item \label{stoch_l2}$\E[u(Y)]\geq \E[u(Z)]$, for all nondecreasing concave utility functions $u$
  \item \label{stoch_l3}$\E[f(Y)]\leq \E[f(Z)]$, for all nondecreasing convex utility functions $f$
  \item \label{stoch_l5}$\E[f(Y)]\leq \E[f(Z)]$, for all convex utility functions $f$ 
  \item \label{stoch_l6}$\E[\max(Y-a,0)]\leq \E[\max(Z-a,0)]$, for all $a\in\mathbb{R}$.
 \end{enumerate}
\end{proposition}

Note that the assumption that $Y$ and $Z$ are nonnegative and have same mean is not required for each of the equivalences listed above;
we refer to~\cite[Chapter 17]{MOA79} for a more detailed discussion of these results.
\ref{stoch_l1}$\Leftrightarrow$\ref{stoch_l2} is proved in the seminal papers by Hadar and Russel~\cite{HR69} and Rothschild and Stiglitz~\cite{RS70}. \ref{stoch_l1}$\Leftrightarrow$\ref{stoch_l3} is shown in Li and Wong~\cite{LW99}, and \ref{stoch_l1}$\Leftrightarrow$\ref{stoch_l4} is due to Atkinson~\cite{Atk70}. \ref{stoch_l4}$\Leftrightarrow$\ref{stoch_l5}$\Leftrightarrow$\ref{stoch_l6} can be found in~\cite{MOA79}.

We recall that in economics, risk aversion is commonly modeled by the fact that risk-averse agents seek to maximize a concave increasing utility function of their wealth, while risk-lovers have a convex increasing utility function. The above proposition tells us that $Y\succeq_{(2)} Z$ means that risk-averse expected-utility maximizers prefer gamble $Y$ over gamble $Z$,
while risk-lovers prefer $Z$ over $Y$, and explains why $Y\succeq_{(2)} Z$ can be interpreted as ``$Y$ is less dispersed than $Z$''.

\smallskip
The Lorenz curve of a nonnegative random variable with finite expectation can be used to define several dispersion indices.
\begin{definition}
	Let $Y$ be a nonnegative random variable with finite expectation. The Gini index $G_Y$ of $Y$ is defined to be twice the area between the line of perfect equality and the Lorenz curve of $Y$, and 
	the Pietra index $P_Y$ is defined to be the maximal distance between the line of perfect equality and the Lorenz curve, i.e.,
	$$
	G_Y := 2 \int_0^1 (p-L_Y(p))\, dp\qquad \text{and}\qquad 
	P_Y:= \max_{p\in[0,1]} p-L_Y(p).
	$$
\end{definition}
Both indices are depicted in Figure~\ref{fig:lorenz}. 
Many other equivalent expressions are known for the above indices, see e.g.~\cite{Eli18}, notably
$$
G_Y=\frac{1}{2\mu}\,\E\big[\,|Y^{(1)}- Y^{(2)}|\,\big] \qquad \text{and}\qquad
P_Y=\frac{1}{2\mu}\,\E\big[\,|Y- \mu|\,\big],
$$
where $\mu:=\E[Y]$ and $Y^{(1)},Y^{(2)}$ are independent copies of $Y$. Thus, it follows from Jensen's inequality that $0\leq P_Y \leq G_Y \leq 1$. 

\bigskip
Returning to SEBP, we will show in the next section that we can obtain improved performance guarantees
for \FLEPT\ when all processing time distributions come from a
certain family. We next introduce the concept of 
stochastically dominated family at the second-order, and show that 
most common families of nonnegative two-parameter probability distributions
satisfy this property when we bound their coefficient of variation.
Recall that the squared coefficient
of variation of a random variable $X$ with mean $\mu$ and variance $\sigma^2$ is
 $\Delta=\frac{\sigma^2}{\mu^2}$.

\begin{definition}
	Let $\mathcal{D}$ be a family of nonnegative random variables
	with finite expectation. 
	We say that
	$\mathcal{D}$ is second-order stochastically dominated (SSD)
	if there exists a nonnegative random variable $Z^{\mathcal{D}}$ such that\vspace{-0.5em} 
	\begin{equation}\label{dominated_family}
	\E[Z^{\mathcal{D}}] = 1
	\qquad \text{and} \qquad 
	\frac{X}{\E[X]} \succeq_{(2)} Z^{\mathcal{D}},\ \ \forall X\in\mathcal{D}.\vspace{-0.5em}
	\end{equation}
	When the above holds, we use the shorthand expression
	``$\mathcal{D}$ is an SSD family with minimal element $Z^{\mathcal{D}}$''. 
\end{definition}

\begin{table}[t]
  \begin{tabular}{cll}
  \hline
  \rule{0pt}{3ex}
   $\mathcal{D}_{\Delta}$ & Family & Minimal element $Z^{\mathcal{D}_{\Delta}}$\\\hline
  Lognormal & \multirow{2}{*}{$\left\{ \mathcal{LN}(\mu,\sigma): \mu\in\R, \sigma^2\leq \log(\Delta + 1)\right\}$}
            & \multirow{2}{*}{$\mathcal{LN}\Big(\log (\frac{1}{\sqrt{\Delta+1}}),\sqrt{\log (\Delta + 1)}\Big)$}\\
  $\mathcal{L}_\Delta$    &\\[0.3em]
  Gamma     & \multirow{2}{*}{$\left\{\operatorname{Gamma}(k,\theta): k\geq\frac{1}{\Delta}, \theta>0\right\}$}
            & \multirow{2}{*}{$\operatorname{Gamma}(\frac{1}{\Delta},\Delta)$}\\
$\mathcal{G}_{\Delta}$ &\\[0.3em]
  Weibull   & \multirow{2}{*}{$\left\{\operatorname{Weibull}(k,\lambda): \frac{\Gamma(1+2/k)}{\Gamma(1+1/k)^2}\leq \Delta+1\right\}$} 
            & \multirow{2}{*}{$\operatorname{Weibull}\left(k_{\Delta}, \frac{1}{\Gamma(1+1/k_\Delta)}\right)$}\\
$\mathcal{W}_\Delta$ & \\[0.3em]
 Uniform    & \multirow{2}{*}{$\left\{\mathcal{U}(a,b): 0\leq a\leq b \leq a+\sqrt{3\Delta}(a+b)\right\}$}
            \ \ \ & \multirow{2}{*}{$\mathcal{U}(1-\sqrt{3\Delta},1+\sqrt{3\Delta})$}\\
$\mathcal{U}_{\Delta}$  &\\[0.3em]
Two-Point & \multirow{2}{*}{$\left\{ a\cdot\operatorname{Bernoulli}(p) : a>0, \frac{1}{\Delta+1}\leq p \leq 1\right\}$}
            & \multirow{2}{*}{$ (\Delta+1)\cdot\operatorname{Bernoulli}(\frac{1}{\Delta+1})$}\\
$\mathcal{B}_{\Delta}$ & \\[0.3em]
$\alpha$-Triangular & \multirow{2}{*}{$\left\{\operatorname{Triang}\Big(a,b,c_{\alpha}\Big): 0\leq a\leq b \leq a + 3 \gamma^\alpha_\Delta\right\}$} 
                   & {\footnotesize $\operatorname{Triang}\Big(1-(1+\alpha)\gamma^\alpha_\Delta,\ 
		1+(2-\alpha)\gamma^\alpha_\Delta,
		$}   \\
$\mathcal{T}^\alpha_\Delta$ & & ~\quad\quad\ \ {\footnotesize $1+(2\alpha-1)\gamma^\alpha_\Delta\Big)$}\\\hline
  \end{tabular}
 \medskip
 \caption{\small\setstretch{1.1} Example of families of nonnegative random variables $\mathcal{D}_{\Delta}$ with squared coefficient of variation bounded by $\Delta$ that are SSD, together with their $\preceq_{(2)}$-minimal element of unit mean $Z^{\mathcal{D}_{\Delta}}$. For Weibull distribution, the coefficient $k_\Delta$
 is the unique positive solution of the equation $\frac{\Gamma(1+2/k)}{\Gamma(1+1/k)^2}= \Delta+1$. The $\alpha$-Triangular distributions
 $\operatorname{Triang}(a,b,c_{\alpha})$
 have a fixed shape parameter $\alpha\in[0,1]$, such that 
 the mode of these distributions is located at
 a fraction $\alpha$ of the support segment $[a,b]$, i.e.,
 $c_{\alpha}=(1-\alpha) a + \alpha b$.
 The $\preceq_{(2)}$-minimal element of these families depends on
 the parameter $\gamma^\alpha_\Delta:=\sqrt{\frac{2\Delta}{\alpha^2-\alpha+1}}.$
 \label{tab:stochdom}}
\end{table}

\begin{proposition}
Let $\mathcal{D}_{\Delta}$ be one of the families of nonnegative random variables with squared coefficient of variation bounded by $\Delta$ listed in Table~\ref{tab:stochdom}. Then, $\mathcal{D}_{\Delta}$ is SSD, with minimal element $Z^{\mathcal{D}_{\Delta}}$ given in the table.
\end{proposition}

\proof Let $\mathcal{D}$ be any of the
families listed in the table. 
It follows from standard formulas that if $X\in\mathcal{D}$
is a random variable with mean $\E[X]=\mu$
and squared coefficient of variation $\mathbb{V}[X]/\mu^2=\delta$, then it holds
$X \overset{d}{=} \mu Z^{\mathcal{D}_\delta}$,
where the equality holds in distribution.
Therefore, to show~\eqref{dominated_family},
it suffices to establish that 
$(Z^{\mathcal{D}_{\delta}})_{\delta \geq 0}$ is monotonically decreasing
for the relation of second-order stochastic dominance, i.e., we have $(\delta \leq \delta' \iff Z_{\delta}^{\mathcal{D}} \succeq_{(2)} Z_{\delta'}^{\mathcal{D}})$. 

For the cases of lognormal distributions,
gamma distributions, and Weibull distributions,
this monotonicity property
is proved in~\cite{Lev73}, \cite{ROS00} and~\cite{LP04}, respectively.
For uniform and two-point distributions, it is straightforward to verify that
$F_{Z^{\mathcal{D}_\delta}}$ and $F_{Z^{\mathcal{D}_{\delta'}}}$ satisfy the single-crossing property~\eqref{singlecrossing}. 
\qed

\begin{remark}
If $X'\overset{d}{=}a+X$ for some $a\geq 0$ and $X\in\mathcal{D}$, where $\mathcal{D}$ is an SSD family with minimal element $Z^{\mathcal{D}}$, then it holds
$\frac{X'}{\E[X']}\succeq_{(2)}\frac{X}{\E[X]}\succeq_{(2)} Z^{\mathcal{D}}$, which can be checked with the single-crossing property~\eqref{singlecrossing}.
 As a result, the families of Table~\ref{tab:stochdom} can be extended by allowing an additional nonnegative location parameter.
 For example, the set of two-point distributions 
 with arbitrary support points $a,b\geq 0$, $\{a+(b-a)\cdot\operatorname{Bernoulli}(p): 0\leq a < b, \frac{1}{\Delta+1}\leq p \leq 1\}$ is SSD with minimal element $Z^{\mathcal{B}_{\Delta}}$.
\end{remark}

We continue with a lemma giving an upper bound 
on the cost of a machine for a fixed assignment policy, when the processing times come from an SSD family. For the case of SEBP, we only need the following result for the function $f:x\mapsto \max(x,1)$, but
we give it in the following form as it holds for a larger class of  functions.

\begin{lemma}~\label{lemm:stoch_dominance_one_machine}
	For an SSD family $\mathcal{D}$ with minimal element $Z^{\mathcal{D}}$, let $P_1,\ldots,P_k\in\mathcal{D}$ with $\sum_{j=1}^k \E[P_j]=x$.
	Then, for every nondecreasing convex function $f$, it holds
	$$
	\mathbb{E}[f(P_1+\ldots+P_k)] \leq \E[f(x Z^{\mathcal{D}})].
	$$
\end{lemma}

\proof
	Denote by $\mu_j$ the expectation
	of $P_j$, so that the second-order dominance
	property~\eqref{dominated_family} reads
	$\frac{P_j}{\mu_j} \succeq Z^{\mathcal{D}}$, for all $j\in[k]$. Then,
	we know from Theorem~10 of the work by Li and Wong~\cite{LW99}
	that $\sum_{j=1}^k P_j \succeq_{(2)} \sum_{j=1}^k \mu_j Z^{(j)}$,
	where $Z^{(1)},\ldots,Z^{(k)}$ 
	are independent copies
	of $Z^{\mathcal{D}}$.
	Now, we can apply Theorem~12 of~\cite{LW99}, which
	states that any convex combination of independent copies
	of some random variable $X$ has second-order stochastic dominance over
	$X$ itself. Hence,
	$$
	\sum_{j=1}^k P_j \succeq_{(2)}
	\sum_{j=1}^k \mu_j Z^{(j)}
	= x \sum_{j=1}^k \frac{\mu_j}{x} Z^{(j)}
	\succeq_{(2)} x Z^{\mathcal{D}}.$$
	Finally, we observe that $\sum_{j=1}^k P_j$ and $x Z^{\mathcal{D}}$ have the same mean ($=x$), hence we obtain the desired result from Proposition~\ref{prop:equiv_stochdom}~\ref{stoch_l3}.\qed

\begin{lemma}\label{lemm:stochmin}
	Let $\mathcal{D}$ be an SSD family with minimal element $Z^{\mathcal{D}}$, and define $g^{\mathcal{D}}(x):=\E[\max(xZ^{\mathcal{D}},1)]$.
	Then, the function $x\mapsto 1+x-g^{\mathcal{D}}(x)$ is nondecreasing, and for all $X\in\mathcal{D}$ it holds
	$$
	1+\E[X]-g^{\mathcal{D}}(\E[X]) \leq \E[\min(X,1)].
	$$
\end{lemma}
\proof
	First, note that $g^{\mathcal{D}}$ is a convex function as the expectation of a convex function, hence,
	its right derivative $g_{\mathcal{D}}'^+(x)$ exists for all $x>0$ and is a nondecreasing function.
	Since $\max(t,1) \leq 1+t$ holds for all $t\geq 0$, we have
	\sloppy{$g^{\mathcal{D}}(x)=\E[\max(xZ^{\mathcal{D}},1)]\leq \E[1+xZ^{\mathcal{D}}]=1+x,\ $} for all $x\geq 0$. This implies that the right derivatives
	of $g^{\mathcal{D}}$ satisfy $g_{\mathcal{D}}'^{+}(x)\leq 1, \forall x>0$,
	hence the function $x\mapsto 1+x-g^{\mathcal{D}}(x)$ is nondecreasing.
	
	Using the equality $1+t=\max(t,1)+\min(t,1)$, we obtain
	\begin{align*}
	1 + \E[X] = \E\big[1+\E[X] Z^{\mathcal{D}}\big]  
	= &\ \underbrace{ \E[\max(\E[X] Z^{\mathcal{D}},1)]}_{=g^{\mathcal{D}}(\E[X])}\ + \ \underbrace{\E[\min(\E[X] Z^{\mathcal{D}},1)]}_{\leq \E[\min(X,1)]},
	\end{align*}
	where the inequality follows from Proposition~\ref{prop:equiv_stochdom}~\ref{stoch_l2}, 
	using the fact that
	$\mathcal{D}$ is an SSD family, so $\E[X] \cdot Z^{\mathcal{D}}\preceq_{(2)} X$, and that
	$t\mapsto \min(t,1)$ is concave nondecreasing.\qed

\medskip
Now, we are going to apply these lemmas in order
to get improved performance guarantees 
for SEBP instances with processing times
in an SSD family.

\begin{theorem}\label{theo:bound_gdelta}
	Let $P_1,\ldots,P_n\in\mathcal{D}$
	for an SSD family
	$\mathcal{D}$ with minimal element $Z^{\mathcal{D}}$,
	and let $g^{\mathcal{D}}(x):= \E[\max(xZ^{\mathcal{D}},1)]$. 
	Then, we have
	$$\frac{\LEPTmu}{\OPT} \leq \sup_{t\in[0,1]} (2-\frac{1}{t}) g^{\mathcal{D}}(t) + (\frac{1}{t}-1) g^{\mathcal{D}}(2t).$$
	Numerical values of this bound for several distribution
	families $\mathcal{D}$  are indicated in Table~\ref{tab:approx}.
\end{theorem}

\proof
Denote by $\ell$ the minimum expected load
of a machine, as in Lemma~\ref{lemm:bound_ell}.
The main work in this proof will be to show that
there exists a subset of machines $\mathcal{M'}$ of cardinality $m'$, and some $x_i'\in[\ell,2\ell]$,  $\forall i\in\mathcal{M}'$, such that
\begin{equation}\label{startratio}
\frac{\LEPTmu}{\OPT} \leq \frac{ \sum_{i\in\mathcal{M}'} g^{\mathcal{D}}(x_i') }{\max(\sum_{i\in\mathcal{M'}} x_i',m')}.
\end{equation}
Let us first prove the theorem assuming that the above claim is valid. 
We express each $x_i$ as a convex combination of $\ell$ and $2\ell$, writing $x_i'=(2-\frac{x_i'}{\ell})\cdot \ell + (\frac{x_i'}{\ell}-1)\cdot 2\ell$. By convexity of $g^{\mathcal{D}}$,
we have $g^{\mathcal{D}}(x_i')\leq (2-\frac{x_i'}{\ell})\cdot g^{\mathcal{D}}(\ell) + (\frac{x_i'}{\ell}-1)\cdot g^{\mathcal{D}}(2\ell)$. Now, let
$s':=\sum_{i\in\mathcal{M}'} x_i'$. Summing over all $i\in\mathcal{M'}$ we have
$$\hspace{-0.5mm}
	\frac{\LEPTmu}{OPT} \leq \frac{ (2m'-\frac{s'}{\ell}) g^{\mathcal{D}}(\ell) +  (\frac{s'}{\ell}-m') g^{\mathcal{D}}(2\ell) }{\max(s',m')} = \frac{ (2-\frac{\rho}{\ell}) g^{\mathcal{D}}(\ell) + (\frac{\rho}{\ell}-1) g^{\mathcal{D}}(2\ell) }{\max(\rho,1)},
$$
where we have set $\rho:=s'/m'$.	We will now show that the above bound is maximized for $\rho=1$.
To see this, notice that
$f:x\mapsto\max(x,1)$ satisfies $f(x)\leq f(2x)\leq 2f(x)$ for any $x\geq 0$. Hence, taking the expectation we obtain
	\sloppy{$g^{\mathcal{D}}(x)\leq g^{\mathcal{D}}(2x)\leq 2g^{\mathcal{D}}(x)$ for any $x\geq 0$.} Then, our claim follows 
	from the derivative of the above bound with respect
	to $\rho$, which is equal
	to $\frac{1}{\ell}( g^{\mathcal{D}}(2\ell)-g^{\mathcal{D}}(\ell))\geq 0$ for all $\rho\leq 1$,
	and is equal to $-\frac{1}{\rho^2} (2g^{\mathcal{D}}(\ell)-g^{\mathcal{D}}(2\ell))\leq 0$ for all $\rho\geq 1$.
	
	To obtain the statement of the lemma, we show that the supremum
	of the function $h^{\mathcal{D}}: \R_{\geq 0}\to\R_{\geq 0}$ with $t\mapsto (2-\frac{1}{t}) g^{\mathcal{D}}(t) + (\frac{1}{t}-1) g^{\mathcal{D}}(2t)$ is attained in the interval $[0,1]$. For $t>1$, we express $t$ as a convex combination of $1$ and $2t$, that is, $t = \frac{t}{2t-1}\cdot 1 + \frac{t-1}{2t-1} \cdot 2t$. We have
	$g^{\mathcal{D}}(t)\leq \frac{t}{2t-1} g^{\mathcal{D}}(1) + \frac{t-1}{2t-1} g^{\mathcal{D}}(2t)$ by convexity of $g^{\mathcal{D}}$. Multiplying both sides of this inequality by $2-\frac{1}{t}$, we obtain
	$$(2-\frac{1}{t}) g^{\mathcal{D}}(t) \leq g^{\mathcal{D}}(1) + (1-\frac{1}{t}) g^{\mathcal{D}}(2t) \implies
	h^{\mathcal{D}}(t)\leq g^{\mathcal{D}}(1)=h^{\mathcal{D}}(1).$$

It remains to show that~\eqref{startratio} holds
for some $\mathcal{M}'\subseteq \mathcal{M}$
of cardinality $m'$ and some $x'\in[\ell,2\ell]^{m'}$.
Applying Lemma~\ref{lemm:stoch_dominance_one_machine} to the function $f:x\mapsto\max(x,1)$,
we obtain $\FLEPT\leq \sum_{i\in\mathcal{M}} g^{\mathcal{D}}(x_i)$. Then, we readily observe that
if $x_i\leq 2\ell$ holds for all machines,
we could simply set $\mathcal{M}'=\mathcal{M}$ and $x_i'=x_i$, $\forall i\in\mathcal{M}$, and~\eqref{startratio} would
follow from Proposition~\ref{prop:order}.
Moreover, we know from
Lemma~\ref{lemm:bound_ell}
that $x_i\leq 2\ell$ holds whenever $n_i\geq 2$. Thus, we only need
to take special care of those machines where \FLEPT\
assigns a single job and $x_i > 2\ell$. To this end, we introduce a partition of the machines relying on the truncated loads $\alpha_i$'s:
\begin{align*}
&\mathcal{M}_0:=\{i\in\mathcal{M}:\ n_i=1,\ x_i>2\ell,\ \alpha_i>2\ell\}\\
&\mathcal{M}_1:=\{i\in\mathcal{M}:\ n_i=1,\ x_i>2\ell,\ \alpha_i\leq 2\ell\}\\
&\mathcal{M}_2:=\mathcal{M}\setminus(\mathcal{M}_0\cup \mathcal{M}_1).
\end{align*}
Define $x'_i:=\max(\ell,\alpha_i)$ for all $i\in\mathcal{M}_0\cup \mathcal{M}_1$,
and $x'_i:=x_i$ for all $i\in\mathcal{M}_2$. Furthermore, we define $\mathcal{M}':=\mathcal{M}\setminus\mathcal{M}_0$ with $m':=|\mathcal{M}'|$,
and we note that $x_i'$ lies in the interval $[\ell,2\ell]$ for all $i\in\mathcal{M}'$, as required.
	
Let us bound the cost induced by \FLEPT\ on each machine. Let $\delta_i:=x_i-x_i'\geq 0$, $\forall i \in\mathcal{M}$. We claim that $\E [\max(X^{\FLEPT}_i,1)]\leq g^{\mathcal{D}}(x_i')+\delta_i$ holds for each machine.
For the machines $i\in\mathcal{M}_2$, this results
from Lemma~\ref{lemm:stoch_dominance_one_machine},
together with $x_i'=x_i$ and $\delta_i=0$.
For the other machines, we have 
$\E [\max(X^{\FLEPT}_i,1)]=1+\beta_i$,
because these machines host a single job.
Now, we distinguish two cases: \texttt{(1)} If $\alpha_i \geq \ell$, then $1+\beta_i\leq g^{\mathcal{D}}(x_i')+\delta_i$
follows from $1\leq g^{\mathcal{D}}(x_i')$
and $\delta_i=x_i-\alpha_i=\beta_i$;
\texttt{(2)} If $\alpha_i<\ell$, then we obtain
from Lemma~\ref{lemm:stochmin} and $x_i>\ell$ that
$$
1+\ell-g^{\mathcal{D}}(\ell) \leq 1+x_i-g^{\mathcal{D}}(x_i) \leq \E[\min(X_i^{\FLEPT},1)] = \alpha_i,
$$
which implies $1+ \beta_i\leq \alpha_i +\beta_i - \ell + g^{\mathcal{D}}(\ell) =g^{\mathcal{D}}(\ell)+x_i-\ell=g^{\mathcal{D}}(x_i')+\delta_i$. Hence, the claim is proved, and
summing the bound over all machines yields
\begin{equation}\label{Uphi}
\FLEPT\leq \sum_{i\in\mathcal{M}} g^{\mathcal{D}}(x_i')+\delta_i.
\end{equation}
On the other hand, we have
\begin{equation}\label{Lopt}
OPT \geq \max(\sum_{i\in\mathcal{M}} x_i,m+\sum_{i\in\mathcal{M}} \beta_i)
\geq \max(\sum_{i\in\mathcal{M}} x_i',m) +\sum_{i\in\mathcal{M}}\delta_i
\end{equation}
where the first inequality is
Lemma~\ref{lemm:LB}, and
the second one follows from $\beta_i\geq \delta_i, \forall i$.
Now, we combine~\eqref{Uphi} and~\eqref{Lopt}, and use the fact that removing $\sum_{i\in\mathcal{M}} \delta_i\geq 0$ from both the numerator and the numerator can only worsen the ratio, to obtain
$$
\frac{\LEPTmu}{OPT} \leq \frac{\sum_{i\in\mathcal{M}} g^{\mathcal{D}}(x_i')}{\max(\sum_{i\in\mathcal{M}} x_i',m)}.
	$$
At this stage, observe that we have already proved that~\eqref{startratio} holds in the case where $\mathcal{M}_0$ is empty.
So it only remains to handle the case $\mathcal{M}_0\neq \emptyset$. In this case, we have $2\ell <\alpha_i \leq 1$ for all $i\in\mathcal{M}_0$, which implies $x_i'\leq 1$ for all machines $i\in\mathcal{M}$. Thus, 
$\sum_{i\in\mathcal{M}} x_i'\leq m$ and $\sum_{i\in\mathcal{M}'} x_i'\leq m'$. Altogether, we obtain
\begin{equation}
 \frac{\LEPTmu}{OPT} \leq
 \frac{\sum_{i\in\mathcal{M}} g^{\mathcal{D}}(x_i')}{m}
 \leq
 \frac{\sum_{i\in\mathcal{M}'} g^{\mathcal{D}}(x_i')}{m'}=
 \frac{\sum_{i\in\mathcal{M}'} g^{\mathcal{D}}(x_i')}{\max(\sum_{i\in\mathcal{M}'} x_i',\ m')},
\end{equation}
where we have used $g^{\mathcal{D}}(x_i')\leq 1$ for all $i\in\mathcal{M}_0$ in the second inequality, so we could remove the constant $|\mathcal{M}_0|=m-m'$ from both the numerator and the denominator in order to increase the ratio. This concludes the proof.
\qed	

\begin{table}
{\small
$$
{\renewcommand{\arraystretch}{1.2}
\begin{array}{l|ccccccc}
\hline
\Delta & 0 & \frac{1}{8}& \frac{1}{6}& \frac{1}{4}& \frac{1}{3}& \frac{1}{2}& 1 \\ \hline
\mathcal{L}_\Delta& ~~1.1716~~& ~~1.1990~~& ~~1.2112~~& ~~1.2334~~& ~~1.2526~~& ~~1.2843~~& ~~1.3485~~\\[1mm]
\mathcal{G}_\Delta& ~~1.1716~~& ~~1.2012~~& ~~1.2148~~& ~~1.2401~~& ~~1.2629~~& ~~1.3023~~& ~~1.3896~~\\[1mm]
\mathcal{W}_\Delta& ~~1.1716~~& ~~1.2044~~& ~~1.2186~~& ~~1.2450~~& ~~1.2685~~& ~~1.3080~~& ~~1.3896~~\\[1mm]
\mathcal{U}_\Delta& ~~1.1716~~& ~~1.2041~~& ~~1.2210~~& ~~1.2526~~& ~~1.2812~~& - & - \\[1mm]
\mathcal{B}_\Delta& ~~1.1716~~& ~~1.2222~~& ~~1.2385~~& ~~1.2702~~& ~~1.3005~~& ~~1.3573~~& ~~1.5000~~\\[1mm]
\mathcal{T}^0_\Delta& ~~1.1716~~& ~~1.2000~~& ~~1.2163~~& ~~1.2468~~& ~~1.2744~~& ~~1.3230~~& - \\[1mm]
\mathcal{T}^{\frac{1}{4}}_\Delta& ~~1.1716~~& ~~1.2012~~& ~~1.2170~~& ~~1.2468~~& - & - & - \\[1mm]
\mathcal{T}^{\frac{1}{2}}_\Delta& ~~1.1716~~& ~~1.2053~~& ~~1.2207~~& - & - & - & - \\[1mm]
\mathcal{T}^{\frac{3}{4}}_\Delta& ~~1.1716~~& ~~1.2079~~& - & - & - & - & - \\[1mm]
\mathcal{T}^1_\Delta& ~~1.1716~~& ~~1.2088~~& - & - & - & - & - \\
\hline
\end{array}
}
$$
}
\caption{\small\setstretch{1.1} Approximation guarantees from Theorem~\ref{theo:bound_gdelta} for several families of probability distributions (cf.\ Table~\ref{tab:stochdom} for a description of the symbols used in the first column) and upper bound $\Delta$ on the squared coefficient of variation. Note that $\Delta\leq \frac{1}{3}$ always holds for a nonnegative uniform random variable,
and $\Delta\leq\frac{\alpha^2-\alpha+1}{2(1+\alpha)^2}$ for a nonnegative $\alpha$-triangular distributed variable.
\label{tab:approx}}
\end{table}

\begin{remark}
	As Table~\ref{tab:approx} shows, the bound of Theorem~\ref{theo:bound_gdelta} is not tight. For example, the bound can get larger than $1+e^{-1}$ for large values of $\Delta$. 
	Moreover, if $\Delta=0$ (corresponding to the deterministic version EBP), the bound equals $4-2\sqrt{2}\simeq1.1716$, but a tight bound of $\frac{13}{12}\simeq 1.0833$ is known in this case~\cite{DKST98}.
\end{remark}

\begin{remark}
	For the case of lognormal processing times ($P_j\in\mathcal{L}_\Delta$, $\forall j$),
	we conjecture\footnote{We could check
		with a symbolic computation software
		that $t=\frac{1}{\sqrt{2}}$ is a local maximum,
		and the function to maximize seems to be unimodal,
		but we didn't invest more time to prove this.}
	that the supremum in the bound of Theorem~\ref{theo:bound_gdelta} is
	always reached at $t = \frac{1}{\sqrt{2}}$. If the
	conjecture is true, we would obtain the following closed
	form formula for an upper bound on the performance guarantee of $\LEPTmu$:
	$$
	\frac{\LEPTmu}{OPT_\mathcal{P}} \leq 
	2(\sqrt{2}-1)
	\left[
	1+\sqrt{2}
	\Phi_0\left(\frac{\ln(2(\Delta + 1))}{2 \sqrt{\ln(\Delta + 1)}} \right)
	-\Phi_0\left(\frac{\ln\frac{2}{\Delta + 1}}{2 \sqrt{\ln(\Delta + 1)}} \right)
	\right],
	$$
	where $\Phi_0$ denotes the cumulative distribution function of the standard normal law.
\end{remark}

We will now use Theorem~\ref{theo:bound_gdelta} to derive a distribution-free
bound that depends only on the Pietra index of the processing times.

\begin{theorem}\label{theo:pietra}
	Let $P_1,\ldots,P_n$ be nonnegative random variables with
	finite expectation and Pietra index at most $\varrho$. Then,
	$$
	\frac{\LEPTmu}{OPT} \leq 2 \left( 2-\sqrt{2} + \varrho(\sqrt{2}-1)\right).
	$$
\end{theorem}

\begin{remark}
	Due to the inequality $P_Y\leq G_Y$, the above result also holds if all processing times have Gini index at most $\varrho$.
\end{remark}

\proof{
	We first prove the result for the case where 
	all random variables have bounded support,
	and the result will follow
	by standard continuity arguments. 
	For a constant $\theta$ large enough (we require 
	$\theta\cdot(1-\varrho)>1$), define
	$\mathcal{D}^\varrho_\theta$
	as the set of all nonnegative random variables with finite expectation
	and Pietra index at most $\varrho$ such that 
	$Y\leq \theta \E[Y]$ holds almost surely.
	Henceforth we assume $P_j\in\mathcal{D}^\varrho_\theta,$ for all jobs $j$.
	
	Our assumption implies that $F_{P_j}^{-1}(1)\leq \theta \E[P_j]$, hence the left derivative of the Lorenz function
	at $p=1$ is $L_{p_j}'^-(1)=\frac{1}{\E[P_j]}F_{P_j}^{-1}(1)\leq \theta$. Using this and the convexity of $L_{P_j}$, we obtain
	$\sloppy{L_{P_j}(p) \geq 1+\theta (p-1)}$ for all $ p\in[0,1]$.
	By definition of the Pietra index, we know for all $p\in[0,1]$ that 
	$L_{P_j}(p)\geq p-\varrho$. So we have
	$$
	L_{P_j}(p) \geq \max(0,p-\varrho,1+\theta (p-1)),\quad \forall p\in[0,1].
	$$
	It is easy to see that the right-hand side of the above
	expression coincides with the Lorenz curve of the random variable
	$Z^\varrho_\theta$ such that
	$$
	\mathbb{P}[Z^\varrho_\theta=0] = \varrho,\qquad \mathbb{P}[Z^\varrho_\theta=1] = 1-\varrho\frac{\theta}{\theta-1},\quad \text{and}\quad \mathbb{P}[Z^\varrho_\theta=\theta]=\frac{\varrho}{\theta-1},
	$$
	where these probabilities are nonnegative since we assumed $\theta \cdot(1-\varrho)>1$.
	The Lorenz curve is scale-invariant by construction,
	so $P_j$ and 
	$\overline{P_j}:=\frac{P_j}{\E[P_j]}$
	have the same Lorenz curve. The above inequality
	indicates that $\overline{P_j} \preceq_{L} Z_\theta^\varrho$, 
	and so it holds $\overline{P_j} \succeq_{(2)} Z_\theta^\varrho$ 
	by Proposition~\ref{prop:equiv_stochdom}~\ref{stoch_l4}.
	This shows that the family $\mathcal{D}^\varrho_\theta$ is
	SSD with minimal element $Z^\varrho_\theta$, so by Theorem~\ref{theo:bound_gdelta}:
	$$\frac{\LEPTmu}{OPT} \leq \sup_{t\in[0,1]} (2-\frac{1}{t}) g_\theta(t) + (\frac{1}{t}-1) g_\theta(2t),$$
	where $g_\theta(t)=\E[\max(1,t\cdot Z^{\varrho}_\theta)]=
	\max(1,\varrho+1+\frac{\varrho\theta}{\theta-1}(t-1),\varrho+t)$. 
	It is easy to see that for all $\sloppy{t>0}$, $g_\theta(t)$ is nondecreasing with respect to $\theta$. As a consequence, we obtain
	$\sloppy{g_\theta(t)\leq g_\infty(t):=\max(1+\varrho t,\varrho+t)}$ for all $ t\geq 0$. Finally, simple calculus shows that the function
	$t\mapsto  (2-\frac{1}{t}) g_\infty(t) + (\frac{1}{t}-1) g_\infty(2t)$ reaches its maximum over $t\in[0,1]$ at 
	$t=\frac{1}{\sqrt{2}}$, and we get the desired result
	after substitution. \qed
	}

\medskip
 This theorem improves the bound of $1+e^{-1}$ from~\cite{SST18} for all instances with a Pietra index bounded by $\varrho \leq \frac{1+\sqrt{2}}{2}e^{-1}+\frac{1-\sqrt{2}}{2}\approx 0.237$.


\section{Conclusion and Future work}

We showed that $\LEPTmu$ is, in some sense, the best algorithm among the class of fixed assignment policies we can hope for. 
This result might inspire future work to consider the same or similar and related ratios for other scheduling problems,
in which we compare within or against several subclasses of policies,
in order to obtain more interesting and precise results on the performance of algorithms.

Moreover, we studied the worst-case behaviour of the $\LEPTmu$ policy
for instances with bounded Pietra index, or for 
second-order stochastically dominated families of random processing times. It would be interesting to investigate whether these techniques can be applied to other stochastic scheduling problems.

Another direction for future work on SEBP
is the study of the case of unequal bins,
which is relevant for the application to surgery scheduling, where
operating rooms may have different opening hours.
Since the class of fixed assignment policies is relevant for surgery scheduling,
another interesting open question is whether there exists a policy $\Pi\in\mathcal{F}$ 
with a performance guarantee $<\frac{4}{3}$ in the class $\mathcal{F}$.

Last but not least, a two-stage stochastic online extension of the EBP could yield a better understanding
of policies for the surgery scheduling problem with add-on cases (emergencies).


\end{document}

%% file: lorenz.pdf_tex
\begingroup%
  \makeatletter%
  \providecommand\color[2][]{%
    \errmessage{(Inkscape) Color is used for the text in Inkscape, but the package 'color.sty' is not loaded}%
    \renewcommand\color[2][]{}%
  }%
  \providecommand\transparent[1]{%
    \errmessage{(Inkscape) Transparency is used (non-zero) for the text in Inkscape, but the package 'transparent.sty' is not loaded}%
    \renewcommand\transparent[1]{}%
  }%
  \providecommand\rotatebox[2]{#2}%
  \newcommand*\fsize{\dimexpr\f@size pt\relax}%
  \newcommand*\lineheight[1]{\fontsize{\fsize}{#1\fsize}\selectfont}%
  \ifx\svgwidth\undefined%
    \setlength{\unitlength}{335.01553421bp}%
    \ifx\svgscale\undefined%
      \relax%
    \else%
      \setlength{\unitlength}{\unitlength * \real{\svgscale}}%
    \fi%
  \else%
    \setlength{\unitlength}{\svgwidth}%
  \fi%
  \global\let\svgwidth\undefined%
  \global\let\svgscale\undefined%
  \makeatother%
  \begin{picture}(1,0.68762367)%
    \lineheight{1}%
    \setlength\tabcolsep{0pt}%
    \put(0,0){\includegraphics[width=\unitlength,page=1]{lorenz.pdf}}%
    \put(0.62401848,0.01779593){\color[rgb]{0,0,0}\makebox(0,0)[lt]{\lineheight{1.25}\smash{\begin{tabular}[t]{l}$p$\end{tabular}}}}%
    \put(0.02917134,0.03168268){\color[rgb]{0,0,0}\makebox(0,0)[lt]{\lineheight{1.25}\smash{\begin{tabular}[t]{l}$0$\end{tabular}}}}%
    \put(0.01578457,0.54927758){\color[rgb]{0,0,0}\makebox(0,0)[lt]{\lineheight{1.25}\smash{\begin{tabular}[t]{l}$1$\end{tabular}}}}%
    \put(0.54703899,0.01594355){\color[rgb]{0,0,0}\makebox(0,0)[lt]{\lineheight{1.25}\smash{\begin{tabular}[t]{l}$1$\end{tabular}}}}%
    \put(0.34337384,0.25764396){\color[rgb]{0,0,0}\makebox(0,0)[lt]{\lineheight{1.25}\smash{\begin{tabular}[t]{l}$\frac{G_Y}{2}$\end{tabular}}}}%
    \put(0.53831141,0.343851){\color[rgb]{0,0,0}\makebox(0,0)[lt]{\lineheight{1.25}\smash{\begin{tabular}[t]{l}$y=L_Y(p)$\end{tabular}}}}%
    \put(0.36861317,0.46370686){\color[rgb]{0,0,0}\makebox(0,0)[lt]{\lineheight{1.25}\smash{\begin{tabular}[t]{l}$y=p$\end{tabular}}}}%
    \put(0.01818235,0.6176205){\color[rgb]{0,0,0}\makebox(0,0)[lt]{\lineheight{1.25}\smash{\begin{tabular}[t]{l}$y$\end{tabular}}}}%
    \put(0,0){\includegraphics[width=\unitlength,page=2]{lorenz.pdf}}%
    \put(0.46020112,0.36081172){\color[rgb]{0,0,0}\makebox(0,0)[lt]{\lineheight{1.25}\smash{\begin{tabular}[t]{l}$P_Y$\end{tabular}}}}%
  \end{picture}%
\endgroup%

%% file: arxiv_version_v2.bbl
\begin{thebibliography}{10}

\bibitem{AAWY98}
N.~Alon, Y.~Azar, G.J. Woeginger, and T.~Yadid.
\newblock Approximation schemes for scheduling on parallel machines.
\newblock {\em Journal of Scheduling}, 1(1):55--66, 1998.

\bibitem{Atk70}
A.B. Atkinson.
\newblock On the measurement of inequality.
\newblock {\em Journal of economic theory}, 2(3):244--263, 1970.

\bibitem{BN15}
N.~Bansal and V.~Nagarajan.
\newblock On the adaptivity gap of stochastic orienteering.
\newblock {\em Mathematical Programming}, 154(1-2):145--172, 2015.

\bibitem{BD17}
B.P. Berg and B.T. Denton.
\newblock Fast approximation methods for online scheduling of outpatient
  procedure centers.
\newblock {\em INFORMS Journal on Computing}, 29(4):631--644, 2017.

\bibitem{BV04}
S.~Boyd and L.~Vandenberghe.
\newblock {\em Convex optimization}.
\newblock Cambridge University Press, 2004.

\bibitem{CI98}
R.~Canetti and S.~Irani.
\newblock Bounding the power of preemption in randomized scheduling.
\newblock {\em SIAM Journal on Computing}, 27(4):993--1015, 1998.

\bibitem{choi2012analysis}
Sangdo Choi and Wilbert~E Wilhelm.
\newblock An analysis of sequencing surgeries with durations that follow the
  lognormal, gamma, or normal distribution.
\newblock {\em IIE Transactions on Healthcare Systems Engineering},
  2(2):156--171, 2012.

\bibitem{CL01}
E.~G. Coffman~Jr. and G.~S. Lueker.
\newblock Approximation algorithms for extensible bin packing.
\newblock In {\em Proceedings of the twelfth annual ACM-SIAM Symposium on
  Discrete Algorithms}, pages 586--588. SIAM, 2001.

\bibitem{CSV12}
J.~R. Correa, M.~Skutella, and J.~Verschae.
\newblock The power of preemption on unrelated machines and applications to
  scheduling orders.
\newblock {\em Mathematics of Operations Research}, 37(2):379--398, 2012.

\bibitem{DGV05}
B.~C. Dean, M.~X. Goemans, and J.~Vondr{\'a}k.
\newblock Adaptivity and approximation for stochastic packing problems.
\newblock In {\em Proceedings of the sixteenth annual ACM-SIAM symposium on
  Discrete algorithms}, pages 395--404. Society for Industrial and Applied
  Mathematics, 2005.

\bibitem{DGV08}
B.~C. Dean, M.~X. Goemans, and J.~Vondr{\'a}k.
\newblock Approximating the stochastic knapsack problem: The benefit of
  adaptivity.
\newblock {\em Mathematics of Operations Research}, 33(4):945--964, 2008.

\bibitem{DKST98}
P.~Dell'Olmo, H.~Kellerer, M.G. Speranza, and Z.~Tuza.
\newblock A 13/12 approximation algorithm for bin packing with extendable bins.
\newblock {\em Information Processing Letters}, 65(5):229--233, 1998.

\bibitem{DS99}
P.~Dell'Olmo and M.G. Speranza.
\newblock Approximation algorithms for partitioning small items in unequal bins
  to minimize the total size.
\newblock {\em Discrete Applied Mathematics}, 94(1-3):181--191, 1999.

\bibitem{DMBH10}
B.T. Denton, A.J. Miller, H.J. Balasubramanian, and T.R. Huschka.
\newblock Optimal allocation of surgery blocks to operating rooms under
  uncertainty.
\newblock {\em Operations research}, 58(4-1):802--816, 2010.

\bibitem{DT02}
F.~Dexter and R.D. Traub.
\newblock How to schedule elective surgical cases into specific operating rooms
  to maximize the efficiency of use of operating room time.
\newblock {\em Anesthesia \& Analgesia}, 94(4):933--942, 2002.

\bibitem{Eli18}
I.~Eliazar.
\newblock A tour of inequality.
\newblock {\em Annals of Physics}, 389:306--332, 2018.

\bibitem{GJ79}
M.~R. Garey and D.~S. Johnson.
\newblock {\em Computers and intractability: a guide to NP-completeness}.
\newblock WH Freeman and Company, San Francisco, 1979.

\bibitem{GNS16}
A.~Gupta, V.~Nagarajan, and S.~Singla.
\newblock Algorithms and adaptivity gaps for stochastic probing.
\newblock In {\em Proceedings of the twenty-seventh annual ACM-SIAM Symposium
  on Discrete Algorithms}, pages 1731--1747. SIAM, 2016.

\bibitem{HR69}
J.~Hadar and W.R. Russell.
\newblock Rules for ordering uncertain prospects.
\newblock {\em The American economic review}, 59(1):25--34, 1969.

\bibitem{ISM10}
D.~Isern, D.~S{\'a}nchez, and A.~Moreno.
\newblock Agents applied in health care: A review.
\newblock {\em International journal of medical informatics}, 79(3):145--166,
  2010.

\bibitem{joustra2013can}
Paul Joustra, Reinier Meester, and Hans van Ophem.
\newblock Can statisticians beat surgeons at the planning of operations?
\newblock {\em Empirical Economics}, 44(3):1697--1718, 2013.

\bibitem{KS07}
S.G. Kolliopoulos and G.~Steiner.
\newblock Approximation algorithms for scheduling problems with a modified
  total weighted tardiness objective.
\newblock {\em Operations research letters}, 35(5):685--692, 2007.

\bibitem{KW02}
M.~Y. Kovalyov and F.~Werner.
\newblock Approximation schemes for scheduling jobs with common due date on
  parallel machines to minimize total tardiness.
\newblock {\em Journal of Heuristics}, 8(4):415--428, 2002.

\bibitem{Leu04}
J.~Y-T. Leung.
\newblock {\em Handbook of scheduling: algorithms, models, and performance
  analysis}.
\newblock CRC Press, 2004.

\bibitem{Lev73}
H.~Levy.
\newblock Stochastic dominance among log-normal prospects.
\newblock {\em International Economic Review}, pages 601--614, 1973.

\bibitem{LW99}
C.-K. Li and W.-K. Wong.
\newblock Extension of stochastic dominance theory to random variables.
\newblock {\em RAIRO-Operations Research}, 33(4):509--524, 1999.

\bibitem{LXCZ09}
M.~Liu, Y.~Xu, C.~Chu, and F.~Zheng.
\newblock Online scheduling to minimize modified total tardiness with an
  availability constraint.
\newblock {\em Theoretical Computer Science}, 410(47-49):5039--5046, 2009.

\bibitem{Lor1905}
M.O. Lorenz.
\newblock Methods of measuring the concentration of wealth.
\newblock {\em Publications of the American Statistical Association},
  9(70):209--219, 1905.

\bibitem{LP04}
M.~Lubrano and C.~Protopopescu.
\newblock Density inference for ranking european research systems in the field
  of economics.
\newblock {\em Journal of Econometrics}, 123(2):345--369, 2004.

\bibitem{MOA79}
A.W. Marshall, I.~Olkin, and B.C. Arnold.
\newblock {\em Inequalities: theory of majorization and its applications}.
\newblock Springer Science \& Business Media, 2011.

\bibitem{MUV06}
N.~Megow, M.~Uetz, and T.~Vredeveld.
\newblock Models and algorithms for stochastic online scheduling.
\newblock {\em Mathematics of Operations Research}, 31(3):513--525, 2006.

\bibitem{MRW84}
R.H. M{\"o}hring, F.J. Radermacher, and G.~Weiss.
\newblock Stochastic scheduling problems {I}—general strategies.
\newblock {\em Zeitschrift f{\"u}r Operations Research}, 28(7):193--260, 1984.

\bibitem{Pin16}
M.~L. Pinedo.
\newblock {\em Scheduling: Theory, Algorithms, and Systems}.
\newblock Springer, 2016.

\bibitem{ROS00}
H.M. Ramos, J.~Ollero, and M.A. Sordo.
\newblock A sufficient condition for generalized lorenz order.
\newblock {\em Journal of Economic Theory}, 90(2):286--292, 2000.

\bibitem{RS70}
M.~Rothschild and J.E. Stiglitz.
\newblock Increasing risk: {I}. a definition.
\newblock {\em Journal of Economic theory}, 2(3):225--243, 1970.

\bibitem{S+18}
G.~Sagnol, C.~Barner, R.~Borndörfer, M.~Grima, M.~Seeling, C.~Spies, and
  K.~Wernecke.
\newblock Robust allocation of operating rooms: A cutting plane approach to
  handle lognormal case durations.
\newblock {\em European Journal of Operational Research}, 271(2):420--435,
  2018.

\bibitem{SST18}
G.~Sagnol, D.~{Schmidt genannt Waldschmidt}, and A.~Tesch.
\newblock The price of fixed assignments in stochastic extensible bin packing.
\newblock In {\em Approximation and Online Algorithms}, pages 327--347.
  Springer International Publishing, 2018.

\bibitem{SS02}
A.~S. Schulz and M.~Skutella.
\newblock Scheduling unrelated machines by randomized rounding.
\newblock {\em SIAM Journal on Discrete Mathematics}, 15(4):450--469, 2002.

\bibitem{SSU16}
M.~Skutella, M.~Sviridenko, and M.~Uetz.
\newblock Unrelated machine scheduling with stochastic processing times.
\newblock {\em Mathematics of operations research}, 41(3):851--864, 2016.

\bibitem{SS14}
A.~J. Soper and V.~A. Strusevich.
\newblock Power of preemption on uniform parallel machines.
\newblock In {\em 17th International Workshop on Approximation Algorithms for
  Combinatorial Optimization Problems (APPROX’14)}, pages 392--402, 2014.

\bibitem{ST99}
M.G. Speranza and Z.~Tuza.
\newblock On-line approximation algorithms for scheduling tasks on identical
  machines with extendable working time.
\newblock {\em Annals of Operations Research}, 86:491--506, 1999.

\bibitem{stepaniak2009modeling}
Pieter~S Stepaniak, Christiaan Heij, Guido~HH Mannaerts, Marcel de~Quelerij,
  and Guus de~Vries.
\newblock Modeling procedure and surgical times for current procedural
  terminology-anesthesia-surgeon combinations and evaluation in terms of
  case-duration prediction and operating room efficiency: a multicenter study.
\newblock {\em Anesthesia \& Analgesia}, 109(4):1232--1245, 2009.

\bibitem{strum1998surgical}
DP~Strum, JH~May, and LG~Vargas.
\newblock Surgical procedure times are well modeled by the lognormal
  distribution.
\newblock {\em Anesthesia \& Analgesia}, 86(2S):47S, 1998.

\bibitem{Uet01}
M.~Uetz.
\newblock {\em Algorithms for deterministic and stochastic scheduling}.
\newblock Cuvillier, 2001.

\bibitem{XJDK18}
G.~Xiao, W.~van Jaarsveld, M.~Dong, and J.~van~de Klundert.
\newblock Models, algorithms and performance analysis for adaptive operating
  room scheduling.
\newblock {\em International Journal of Production Research}, 56(4):1389--1413,
  2018.

\bibitem{ZXG14}
Z.~Zhang, X.~Xie, and N.~Geng.
\newblock Dynamic surgery assignment of multiple operating rooms with planned
  surgeon arrival times.
\newblock {\em IEEE transactions on automation science and engineering},
  11(3):680--691, 2014.

\bibitem{ZYLLSLY16}
M.~Zhu, Z.~Yang, X.~Liang, X.~Lu, G.~Sahota, R.~Liu, and L.~Yi.
\newblock Managerial decision-making for daily case allocation scheduling and
  the impact on perioperative quality assurance.
\newblock {\em Translational perioperative and pain medicine}, 1(4):20, 2016.

\end{thebibliography}
